\documentclass[reprint,aps,pra]{revtex4-2}

\usepackage[utf8]{inputenc}
\usepackage{amsmath, amssymb, amsthm, verbatim, bbm,amsfonts, hyperref, mathrsfs}
\usepackage{complexity}
\usepackage{pifont}
\usepackage{bm}

\usepackage{caption}

\makeatletter
\newcommand{\subalign}[1]{%
	\vcenter{%
		\Let@ \restore@math@cr \default@tag
		\baselineskip\fontdimen10 \scriptfont\tw@
		\advance\baselineskip\fontdimen12 \scriptfont\tw@
		\lineskip\thr@@\fontdimen8 \scriptfont\thr@@
		\lineskiplimit\lineskip
		\ialign{\hfil$\m@th\scriptstyle##$&$\m@th\scriptstyle{}##$\hfil\crcr
			#1\crcr
		}%
	}%
}
\makeatother

\usepackage{amssymb,latexsym}

\newcommand{\eqdef}{=}

\newcommand{\GRS}{G^{\mathrm{rel}}_{\mathrm{Stern}}}

\newcommand{\comm}{y}

\newcommand{\vvv}{\mathbf{v}}

\newcommand{\ev}{{\mathbf{e}}}

\newcommand{\sv}{{\mathbf{s}}}

\newcommand{\vv}{{\mathbf{v}}}

\newcommand{\tv}{{\mathbf{t}}}

\newcommand{\zv}{{\mathbf{0}}}

\newcommand{\Hm}{{\mathbf{H}}}

\newcommand{\norm}[1]{\left\lVert#1\right\rVert}

\newcommand{\pp}{\mathcal{P}}
\renewcommand{\vv}{\mathcal{V}}

\newcommand{\F}{\mathbb{F}}

\newcommand{\Unif}{\xleftarrow{\$}}

\newcommand{\SD}{\ensuremath{\mathrm{SD}}}

\usepackage[dvipsnames]{xcolor}
\definecolor{coll}{HTML}{000090}

\usepackage{graphicx}
\usepackage{amsfonts}
\usepackage{amssymb}
\usepackage{amsmath}

\usepackage{framed}
\usepackage{latexsym}
\usepackage{breakcites}
\usepackage{color}
\usepackage{url}
\usepackage{tikz}
\usepackage{algorithm}
\usepackage{caption}
\usepackage[noend]{algpseudocode}
\usepackage{enumerate}
\usepackage{stmaryrd}
\usepackage{float}

\usepackage{thm-restate}

\usepackage{caption}
\usepackage{subcaption}
\usepackage{color}

\usepackage{tikz}
\def\checkmark{\tikz\fill[scale=0.4](0,.35) -- (.25,0) -- (1,.7) -- (.25,.15) -- cycle;} 

\newtheorem{theorem}{Theorem}

\newtheorem{proposition}{Proposition}

\newtheorem{lemma}{Lemma}

\newtheorem{assumption*}{Assumption}

\newtheorem{problem}{Problem}

\newcommand{\altketbra}[1]{\ketbra{#1}{#1}}

\newcommand{\FQ}{\mathbf{\mathsf{\mathbbm{F}_Q}}}
\renewcommand{\FP}{\mathbf{\mathsf{\mathbbm{F}_P}}}
\newcommand{\kb}[1]{\altketbra{#1}}

\renewcommand{\E}{\mathbb{E}}
\newcommand{\etal}{\emph{et al.~}}

\newcommand{\ie}{\textit{i.e.}}

\newcommand{\bb}{\mathscr{B}}

\newcommand{\TS}{T_{\mathrm{Shift}}}
\newcommand{\TPO}{T^{\mathrm{Phase}1}}
\newcommand{\TPT}{T^{\mathrm{Phase}2}}

\newcommand{\COMMENT}[1]{}

\newcommand{\ket}[1]{|#1\rangle}

\newcommand{\ketbra}[2]{|#1\rangle\langle#2|}
\def\01{\{0,1\}}
\newcommand{\braket}[2]{\langle{#1}|{#2}\rangle} 
\def\01{\{0,1\}}

\newcommand{\eps}{\varepsilon}

\newcommand{\zo}{\{0,1\}}

\newcommand{\cadre}[1]
{
	\begin{tabular}{|p{8.5cm}|}
		\hline
		\vspace*{-0.4cm} #1  \\
		\hline
	\end{tabular}
}

\begin{document}
	\title{Relativistic zero-knowledge protocol for $\NP$ over the internet unconditionally secure against quantum adversaries}
	
	\author{André Chailloux}
	\email{andre.chailloux@inria.fr}
	\affiliation{Inria de Paris, EPI COSMIQ}
	
	\author{Yann Barsamian}
	\email{yann@barsamian.fr}
	\affiliation{Inria de Paris, EPI COSMIQ}
	
	\begin{abstract}
		Relativistic cryptography is a proposal for achieving unconditional security that exploits the fact that no information carrier can travel faster than the speed of light. It is based on space-time constraints but doesn't require quantum hardware. Nevertheless, it was unclear whether this proposal is realistic or not. Recently, Alikhani \textit{et al.}~\cite{ABC+21} performed an implementation of a relativistic zero-knowledge for NP. Their implemented scheme shows the feasibility of relativistic cryptography but it is only secure against classical adversaries.
		In this work, we present a new relativistic protocol for $\NP$ which is secure against quantum adversaries and which is efficient enough so that it can be implemented on everyday laptops and internet connections. We use Stern's zero-knowledge scheme for the Syndrome Decoding problem, which was used before in post-quantum cryptography. The main technical contribution is a generalization of the consecutive measurement framework of~\cite{CL17} to prove the security of our scheme against quantum adversaries, and we perform an implementation that demonstrates the feasibility and efficiency of our proposed scheme. 
	\end{abstract}
	
	\maketitle
\section{Introduction}
\paragraph{Context.}
There is a strong conceptual and practical appeal for building cryptographic schemes which have unconditional security meaning that they cannot be attacked by any classical or quantum computer, even with unlimited computing power. Quantum Key Distribution~\cite{BB84} is a prime example of this, and a huge amount of work has been done to understand its security and perform efficient implementations. Relativistic cryptography is another proposal for achieving unconditional security that exploits the no superluminal signaling (NSS) principle. NSS states that no information carrier can travel faster than the speed of light. The interest of relativistic cryptography is that it can perform coin flipping and bit commitment protocols with unconditional security which is known to be impossible even using quantum information~\cite{May97,LC97}, so relativistic cryptography and quantum cryptography complete each other well for protocols with unconditional security. In order to perform protocols in relativistic cryptography, there has to be some strict space-time constraints between the different agents performing the protocol but they can be done without quantum hardware. The goal of this work is to show the practicality of relativistic cryptography by presenting a new relativistic zero-knowledge protocol for $\NP$ and demonstrating its feasibility in real-life conditions, on standard laptops using a standard internet connection. This is the first time a protocol for relativistic cryptography is implemented in this setting and shows these are much simpler to perform than what we could have expected. 

\paragraph{Relativistic cryptography.}

The idea of using the NSS principle for cryptographic protocols started in a  work by Kent~\cite{Kent99} as a way to physically enforce a no-communication constraint between different agents (a similar idea already existed in multi-prover interactive proofs\cite{BGKW88}, but without any explicit implementation proposal). The original goal of Kent was to bypass the no-go theorems for quantum bit-commitment, and there has been many proposals for unconditionally secure relativistic bit commitment\cite{Kent11,Kent12,KTH13}. The original idea of \cite{BGKW88} was also revisited by Cr\'epeau \etal in \cite{CSST11}. Based on this work, Lunghi \textit{et al.}~devised a relativistic bit commitment protocol involving only four agents, two for Alice and two for Bob \cite{LKB+15} - hereafter called the $\FQ$ relativistic bit commitment. Their protocol is secure against quantum adversaries and a multi-round variant, with longer duration time, was shown to be secure against classical adversaries \cite{LKB+15,CCL15,FF15}. While these protocols only seemed of theoretical interest at first, recent implementations have convincingly demonstrated that the required timing and location constraints can be efficiently enforced. In \cite{VMH+16}, the authors performed a $24$-hour-long bit commitment with the pairs of agents standing $8$km apart.
 \paragraph{Relativistic zero-knowledge protocols for $\NP$-complete problems.}
 One important application of commitment schemes is zero-knowledge protocols. It was first observed in~\cite{CL17} that one can use the single round $\FQ$ relativistic bit commitment scheme to construct a relativistic zero-knowledge protocol for the Hamiltonian cycle problem, which is $\NP$-complete with unconditional security even against entangled adversaries. The communication cost of this protocol however becomes quickly prohibitive and the necessary space-time constraints can't be ensured. A more recent proposal~\cite{CMS+20} constructs a variation over the standard $3$ coloring zero-knowledge protocol. They manage to drastically reduce the communication at each round. However, the number of repetitions required is quite large to obtain classical security and is prohibitively too large to obtain security against quantum adversaries. This proposal (the variant secure against classical adversaries) was recently implemented using some dedicated hardware \cite{ABC+21}.

In this letter, we present a new proposal for relativistic zero-knowledge for $\NP$ based on the Syndrome Decoding problem. This is an $\NP$-complete problem which is also believed to be hard against quantum computer for random instances. We will use here Stern's zero-knowledge scheme~\cite{Ste93} which was used before for post-quantum signature schemes with the $\FQ$ relativistic string commitment.
This protocol will have a moderate amount of communication, but also a small amount of rounds to decrease the soundness error. A comparison between the different schemes is presented in Table \ref{Table1}. 

\begin{table}
	\footnotesize{
		\begin{center}
			\begin{tabular}{|l|c|c|c|c|}
				\hline
				& \#Bytes/Round & \#Repetitions & \# Provers & Quantum Sec \\
				\hline
				\cite{CL17} & $1.89$~MB & 100  & 2 & \checkmark \\
				\cite{ABC+21} & $2$~B & $10^{6}$ & 2 & $\times$ \\
				\cite{ABC+21} & $2$~B & $10^{19}$ & 3 & $\checkmark$ \\
				This work & $17.03$~KB & 340 & 2 & \checkmark  \\
				\hline 
			\end{tabular}
			\caption{Parameters for different zero-knowledge proposals for $100$ security bits (see Appendix \ref{Appendix:Comparison} for more details).}
			\label{Table1}
	\end{center}
}
\end{table}

Our main technical contribution is to prove the security of this protocol against quantum adversaries. In order to do so, we relate its security to an entangled game and prove a lower bound on this game using a new quantum learning lemma on consecutive quantum measurements, in the similar vein of \cite{CL17}. We then implement the key steps of this protocol and show it is efficient enough so that the space-time constraints can be satisfied using standard computers and a standard internet connection. The only specific hardware we require are synchronized clocks.

\section{Preliminaries}
\paragraph{The relativistic $\FQ$ string commitment.}
We recall here the relativistic $\FQ$ string commitment that we will use in our relativistic zero-knowledge protocol. 
We consider a prover $\pp$ that wants to commit a string $z \in \FP$ to a verifier $\vv$. Both parties, have agents respectively $\mathcal{P}_1, \mathcal{P}_2$ and $\mathcal{V}_1, \mathcal{V}_2$. $\pp_1,\vv_1$ are at a certain spatial location and $\pp_2,\vv_2$ at a different spatial location, at a distance $D$ from the first one. The committed string is in $\FP$ and we also consider a set $\FQ \supseteq \FP$, where $Q$ is a parameter of the commitment scheme. 
The protocol (when followed by honest players) consists of 3 phases: preparation, commit,and reveal. The string commitment protocol goes as follows. 
\begin{enumerate}
	\item \emph{Preparation phase}: $\pp_1,\pp_2$ (resp.~$\vv_1,\vv_2$) share a random number $a\in \FQ$ (resp.~$b \in \FQ$).
	\item \emph{Commit phase}: $\vv_1$ sends $b$ to $\pp_1$, who immediately returns $y = a + z *b$ where $z \in \FP$ is the committed string. We map $z \in \FP$ as an element of $\FQ$, since $\FP \subseteq \FQ$, the operations + and * used are those of $\FQ$.
	\item \emph{Reveal phase}: $\pp_2$ reveals the values of $z$ and $a$ to $\vv_2$ who checks that $y = a + z*b$.
\end{enumerate}
This protocol has the following timing properties: let $\tau_1$ the time when $\vv_1$ sends $b$ and $\tau_2$ the time when $\vv_2$ receives $(z,a)$. If $\tau_2 - \tau_1 < D/c$ where $c$ is the speed of light then the NSS principle ensures that the message $(z,a)$ is independent of $b$. The following security properties were proven in~\cite{LKB+15}: 
\begin{itemize}
	\item It is \emph{perfectly hiding}: the verifiers don't have any information about $z$ after the commit phase.
	\item It is \emph{binding}: informally, the provers can change their mind about $z$ after the commit phase only with vanishingly small probability. 
\end{itemize}
 
\paragraph{The syndrome decoding problem.}
The Hamming weight $|\vvv|_H$ of a binary  vector is the number of $1$ coordinates of this vector.

\begin{problem}[Syndrome Decoding - \textup{SD}$(n,k,w)$]\label{prob:decoGenR}	
	$ \ $ 
	\begin{itemize}		\setlength\itemsep{-0.5em}
		\item \textup{Instance:} a matrix $\Hm \in \zo^{(n-k)\times n}$, a column vector $\sv \in \zo^{n-k}$,
		\item \textup{Goal:} output a column vector $\ev\in \zo^n$ such that $\Hm\ev= \sv$ and  $|\ev|_H = w$.
	\end{itemize}
\end{problem}

The Syndrome Decoding problem is $\NP$-complete and also believed to be hard on random instances even against quantum computers. It is the canonical hard problem for code-based cryptography. In order to construct a zero-knowledge protocol for this scheme, we first have to split the instances of our problem into Yes instances and No instances. For the $\SD(n,k,w)$ problem, Yes instances are the pairs $(\Hm,\sv)$ such that a solution ({\ie} a vector $\ev \in \zo^n$ st. ${\Hm}\ev = \sv$ and  $|\ev|_H = w$) exists. No instances are the pairs $(\Hm,\sv)$ where no such solution exists.

\section{Our proposal for relativistic zero-knowledge for $\NP$}
\subsection{Brief definition of a zero-knowledge scheme}

In a zero-knowledge protocol between a prover $\pp$ and verifier $\vv$, they are given an instance of a computational problem which is either a Yes or a No instance. $\pp$ wants to convince $\vv$ that they are in a Yes instance but he doesn't want to reveal any other information to $\vv$. Zero-knowledge protocols have many applications in cryptography, for example for identification schemes. If we start from a Yes instance and both players are honest then $\vv$ should be convinced and always accept (\emph{Completeness}). If $\pp$ is honest then $\vv$ shouldn't learn anything more from its interaction with $\pp$ than the fact that they have a Yes instance (\emph{Zero-knowledge}). If we start from a No instance and for any cheating prover $\pp$, $\vv$ should reject with high probability (\emph{Soundness}).  The honest prover could in theory be computationally unbounded, but in our case, we only require a polynomial time prover, which additionally knows a solution to the problem for Yes instances. We give him this solution in an Auxiliary input. However, cheating provers stay computationally unbounded. 

\subsection{Description of our $1$-round relativistic zero-knowledge protocol for $\NP$}
We combine the $\FQ$ relativistic string commitment and Stern's $1$-round zero-knowledge protocol for $\SD$ in order to get our $1$-round relativistic zero-knowledge protocol for $\SD$. Again, $\pp$ and $\vv$ are split into $2$ agents $\pp_1,\pp_2$ and $\vv_1,\vv_2$. In the honest case, we require $\vv_1,\vv_2$ to be at some distance $D$. We present this protocol in Figure \ref{Figure:1round}. This description is self-contained but we discuss more in length this protocol as well as Stern's original zero-knowledge protocol in Appendix \ref{Appendix:Stern}, which can be a good start for those not familiar with the scheme. The timing constraints ensure that for each $i \in \{1,2\}$, the message sent by  $\pp_i$ is independent of the message sent from $\vv_j$ to $\pp_j$ for $j \neq i$.

\begin{figure}[!h]
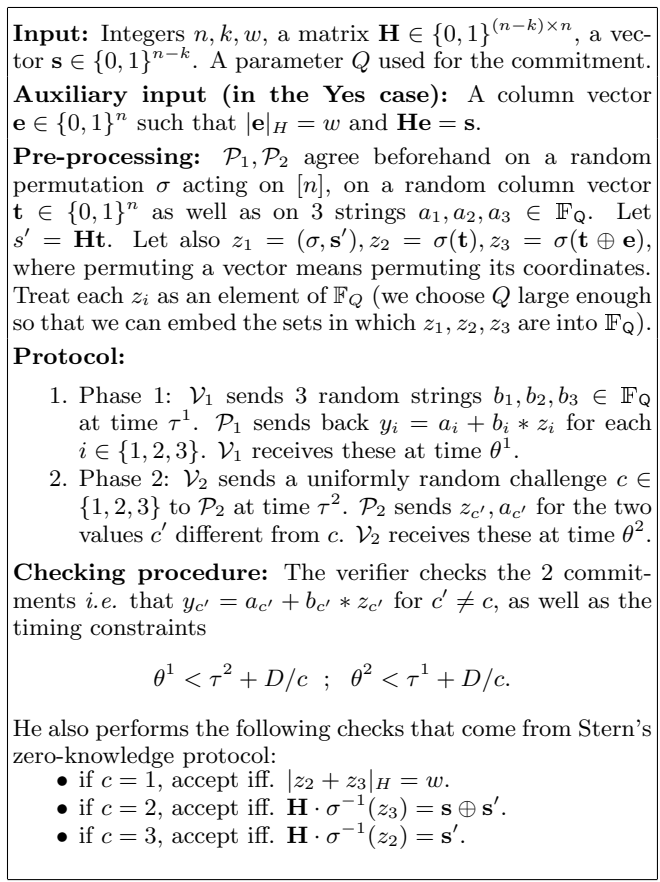

	\cadre{\vspace*{0.2cm}
	\textbf{Input:} Integers $n,k,w$, a matrix $\Hm \in \zo^{(n-k) \times n}$, a vector $\sv \in \zo^{n-k}$. A parameter $Q$ used for the commitment. \\
	\textbf{Auxiliary input (in the Yes case):} A column vector $\ev \in \zo^n$ such that $|\ev|_H = w$ and $\Hm  \ev = \sv$.\\
	\textbf{Pre-processing:} $\pp_1,\pp_2$ agree beforehand on a random permutation $\sigma$ acting on $[n]$, on a random column vector $\tv \in \zo^n$ as well as on $3$ strings $a_1,a_2,a_3 \in \FQ$. Let $s' = \Hm  \tv$. Let also $z_1 = (\sigma ,\sv'), z_2 = \sigma(\tv), z_3 = \sigma(\tv \oplus \ev)$, where permuting a vector means permuting its coordinates. Treat each $z_i$ as an element of $\F_{Q}$ (we choose $Q$ large enough so that we can embed the sets in which $z_1,z_2,z_3$ are into $\FQ$). \\
	\textbf{Protocol:}
	\begin{enumerate}\setlength\itemsep{-0.2em}
		\item  Phase $1$: $\vv_1$ sends $3$ random strings $b_1,b_2,b_3 \in \FQ$ at  time $\tau^1$. $\pp_1$ sends back $y_i = a_i + b_i*z_i$ for each $i \in \{1,2,3\}$. $\vv_1$ receives these at time $\theta^1$.
		\item Phase $2$: $\vv_2$ sends a uniformly random challenge $c \in \{1,2,3\}$ to $\pp_2$ at time $\tau^2$. $\pp_2$ sends $z_{c'},a_{c'}$ for the two values $c'$ different from  $c$. $\vv_2$ receives these at  time $\theta^2$.
	\end{enumerate}
\textbf{Checking procedure:} The verifier checks the $2$ commitments {\ie } that $y_{c'} = a_{c'} + b_{c'}*z_{c'}$ for $c' \neq c$, as well as the timing constraints $$\theta^1 < \tau^2 + D/c \ \ ; \ \ \theta^2 < \tau^1 + D/c.$$ He also performs the following checks that come from Stern's zero-knowledge protocol:
		\vspace*{-0.2cm}
		\begin{itemize}\setlength\itemsep{-0.2em}
			\item if $c = 1$, accept iff. $|z_2 + z_3|_H = w$.
			\item if $c=2$,  accept iff. ${\Hm} \cdot \sigma^{-1}(z_3) = \sv \oplus \sv'.$ 
			\item if $c=3$, accept iff.  ${\Hm} \cdot \sigma^{-1}(z_2) = \sv'.$
		\end{itemize}
}
\caption{$1$-round Relativistic zero-knowledge protocol for {\SD} using the $\FQ$ commitment scheme. }
\label{Figure:1round}
\end{figure}

We prove the security of this scheme. Completeness and the zero-knowledge property follow quite directly from the security of Stern's signature scheme and of the $\FQ$ commitment scheme. The main technical contribution of this work is to bound the soundness of this protocol. We prove the following:

\begin{theorem}
This protocol has perfect completeness, perfect zero-knowledge and has soundness $\frac{2}{3} + \left(\frac{n!2^{4n}}{Q}\right)^{1/4}$ if the space-time constraints are satisfied.
\end{theorem} 
This means that for No instances, an all powerful cheating prover can convince the verifier wp. at most $\frac{2}{3} + \left(\frac{n!2^{4n}}{Q}\right)^{1/4}$.
By taking $Q = 10^{12} n! 2^{4n}$, the soundness becomes $\frac{2}{3} + 0.001$. This seems like a very large $Q$ but recall that sending an element of $\F_Q$ requires $\log_2(Q)$ and performing additions and multiplications in a field of this size is still very efficient. We give the full proof of this Theorem in Appendix~\ref{Apendix:SecurityProof}.

\COMMENT{The proof of Proposition \ref{Proposition:Soundness}, but we will sketch here the main ideas. In the NO case, an adversary cannot answer the $3$ challenges simultaneously in Stern's zero-knowledge protocol without breaking the commitment scheme (see Proposition \ref{Proposition:3ss}).

Equivalently, we can see that for a fixed question/answer pair given to $P_1$, $P_2$ cannot \emph{simultaneously} answer the $3$ challenges without being able to decommit to $2$ different values for the same commitment. From the property of the $\F_Q$ string commitment, this can happen with probability at most $\frac{1}{Q}$.

However, this doesn't necessarily say anything about the entangled value of  $G^Q_{\textrm{Stern-Rel}}$. It could be that when sharing entanglement, the players can win the game with probability $1$ even though they can't answer all the questions simultaneously. Such a paradoxical feature appears for example in the Magic Square game\cite{TODO}.}

\section{Full protocol and implementation}
Our full loss-tolerant relativistic zero-knowledge protocol for $\NP$ is described in Figure \ref{Figure:FullProtocol}. We repeat our $1$-round protocol $R$ times sequentially and allow for a $\lambda$ fraction of rounds where the space-time constraints are not satisfied, for eg. because of losses in the signal. We  extend our security proof to this full protocol in Appendix~\ref{Appendix:FullProtocol}.

\paragraph{Timing constraints.}
We added an extra parameter $T_{shift}$ that will make the space-time constraints easier to satisfy. 
For  round $i$, let $T^{\textrm{Phase}1}_i \eqdef \theta^1_i - \tau^1_i$ and $T^{\textrm{Phase}2}_i \eqdef \theta^2_i - \tau^2_i$. In phase $1$, $\vv_1$ sends $3$ strings in $\FQ$, $\pp_1$ does a computation and sends back $3$ strings in $\FQ$. In phase $2$, $\vv_2$ sends a challenge in $\{1,2,3\}$ and gets back $2$ messages in $\FQ$. This explains why phase $1$ is longer than phase $2$. The timing constraints become for each $i$:
\begin{align}
	\theta^1_i < \tau^2_i + D/c \quad \Rightarrow \quad T^{\textrm{Phase}1}_i - \TS < D/c \\
	\theta^2_i < \tau^1_i + D/c \quad \Rightarrow \quad T^{\textrm{Phase}2}_i + \TS < D/c 
\end{align}
Here, we see why we use $\TS$. Since the $2$ phases take different times, the first constraint would be harder to achieve than the second one with $\TS = 0$. By taking $\TS$ to be an estimate of $\frac{1}{2}\left(\TPO_i - \TPT_i\right)$ for an average $i$, we make the two constraints essentially equally hard to satisfy.

\begin{figure}[!t]
\cadre{
	\vspace*{0.2cm}
\textbf{Parameters:} (n,k,w) for the SD problem. A parameter $Q$ for the commitment used. A parameter $D$ gives the distance between the $2$ verifiers, and a time parameter $\Delta_T$ to delimit the time of a round, a time parameter $\TS$ to determine the time shift between the $2$ phases of the protocol.  A number of rounds $R$ and an allowed fraction of losses $\lambda$.

\begin{enumerate}
\item The $2$ provers and verifiers agree together on an initial time $T_1$ on which they start the protocol. 
\item For i from $1$ to $R$: run the $1$-round relativistic ZK protocol with the $\FQ$ commitment scheme. $\vv_1$ sends his first message at time $\tau^1_i \eqdef  T_1 + (i-1)*\Delta_T$, and $\vv_2$ sends his first message at time $\tau^2_i \eqdef  T_1 + (i-1)*\Delta_T + \TS$. Let $\theta^1_i$ the time at which $\vv_1$ receives the message from $\pp_1$ and  $\theta^2_i$ the time at which $\vv_2$ receives the message from $\pp_2$.
\item At the end of the protocol, the verifiers check the space-time constraints for each $i$ from $1$ to $R$, {\ie} check that $\theta^1_i < \tau^2_i + D/c$ and $\theta^2_i < \tau^1_i + D/c$. Let $F$ be the number of rounds where these space-time constraints are not satisfied. 
\item The verifiers accept if they accept each iteration of the zero-knowledge protocol when the space-time constraints were satisfied and if $F \le \lceil \lambda R\rceil$.
\end{enumerate}
}
\caption{Full loss-tolerant relativistic zero-knowledge protocol for $\NP$}
\label{Figure:FullProtocol}
\end{figure}
\paragraph{Our two scenarios.}

We perform a demonstration of this full scheme using only regular laptops as well as standard network links (ethernet or wifi). We run the experiment in $2$ different scenarios.
\begin{enumerate}
	\item $\vv_1$ and $\pp_1$ are in the same room and are connected through a direct ethernet cable. $\vv_2$ and $\pp_2$ are in a different location but also connected through an ethernet cable. The distance between $\vv_1$ and $\vv_2$ is about $400km$
	\item $\vv_1$ and $\pp_1$ (resp. $\vv_2,\pp_2$) are in different cities and communicate through the usual internet. For each $i$, $\vv_i,\pp_i$ are about $400$km away. We put $\vv_1,\vv_2$ at distance about $9000km$.
\end{enumerate}

These scenarios are illustrated by the following, with examples of cities for which these constraints are satisfied, see Figure \ref{Figure:Scenario12}.

\begin{figure}[!h]
\includegraphics[width = 8cm]{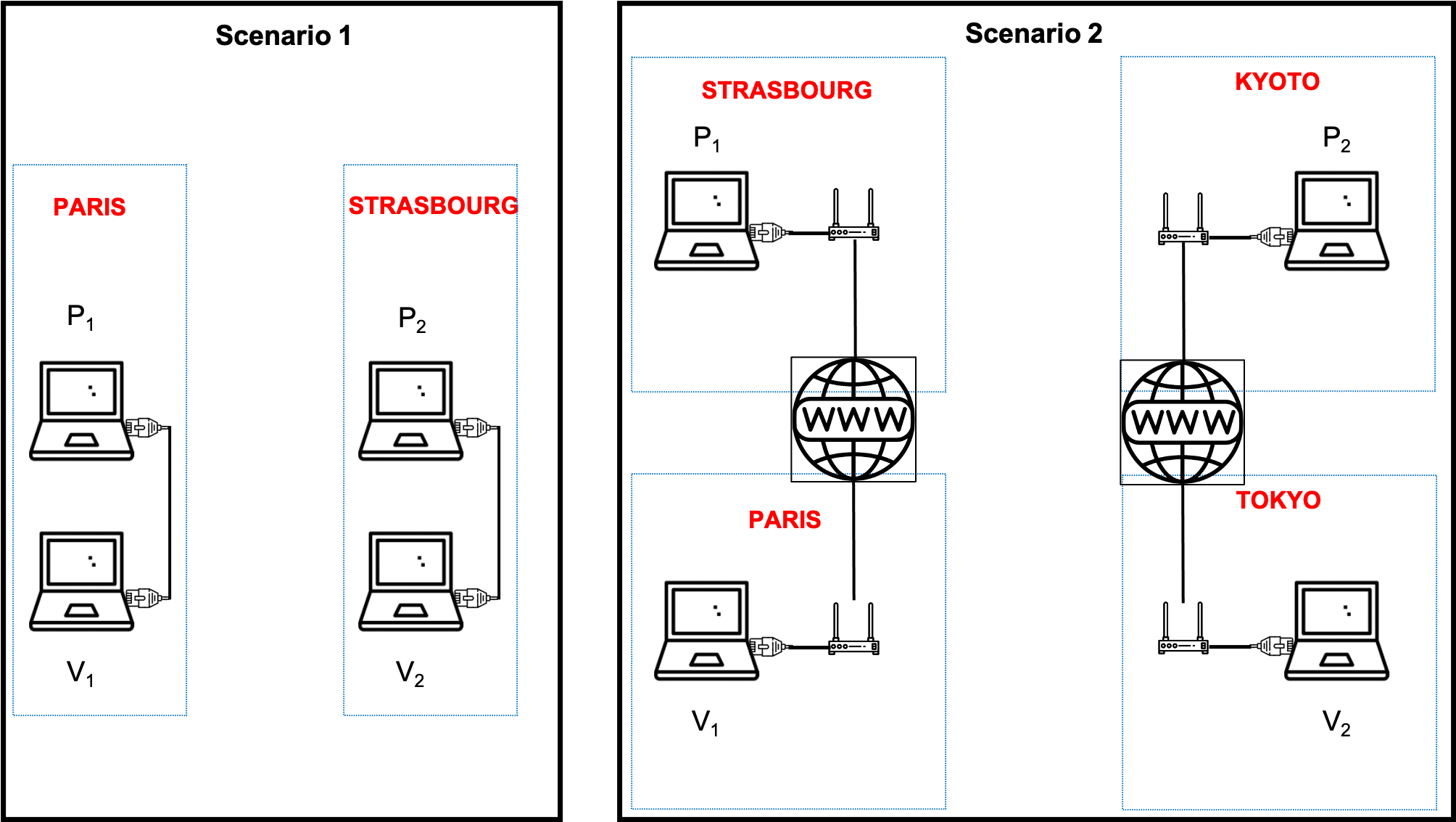}
\caption{Scenarios that we consider for which we demonstrate the feasibility of our full relativistic zero-knowledge protocol for $\NP$.}
\label{Figure:Scenario12}
\end{figure} 

\vspace*{-0.1cm}
\paragraph{Specific implementation parameters.}
Our main protocol that achieves $100$ bits of quantum security has the following parameters that appear in the two scenarios.
\begin{align*}
n & = 1704, \ k = 769, w = 216 \\ R & = 340, F = {22}, \ Q = 2^{23209} - 1, \  \omega^*(\GRS)  \le \frac{2}{3} + 2^{-138} 
\end{align*}
Let $D$ be the distance between $\vv_1$ and $\vv_2$ and let $D'$ be the distance between $\vv_1,\pp_1$ (and also between $\vv_2,\pp_2$). Depending on our scenario, we have the following parameters; where $c \approx 299.8km/sec$ in the speed of light in vacuum. 
\begin{enumerate}
	\item Scenario 1: $D = 400km$, $D' = 10m$,  $D/c \approx 1.33ms$, $\ \Delta_T = 2ms$, $\ T_{shift} = 0.5ms$. With these parameters, the space-time constraints are satisfied for $T^{Phase1}_i < 1.83ms$ and $T^{Phase2}_i < 0.83ms$.
	\item Scenario 2: $D = 9000km$, $D' = 400km$, $D/c \approx 30ms$, $\Delta_T = 40ms$, $T_{shift} = 2.5ms$. With these parameters, the space-time constraints are satisfied for $T^{Phase1}_i < 32.5ms$ and $T^{Phase2}_i < 27.5ms$.
\end{enumerate}

We show in Figures \ref{Figure:TimingsScenario1} and \ref{Figure:TimingsScenario2} the real running times of the different phases. 

\begin{figure}[!t]
	\includegraphics[width = 8cm]{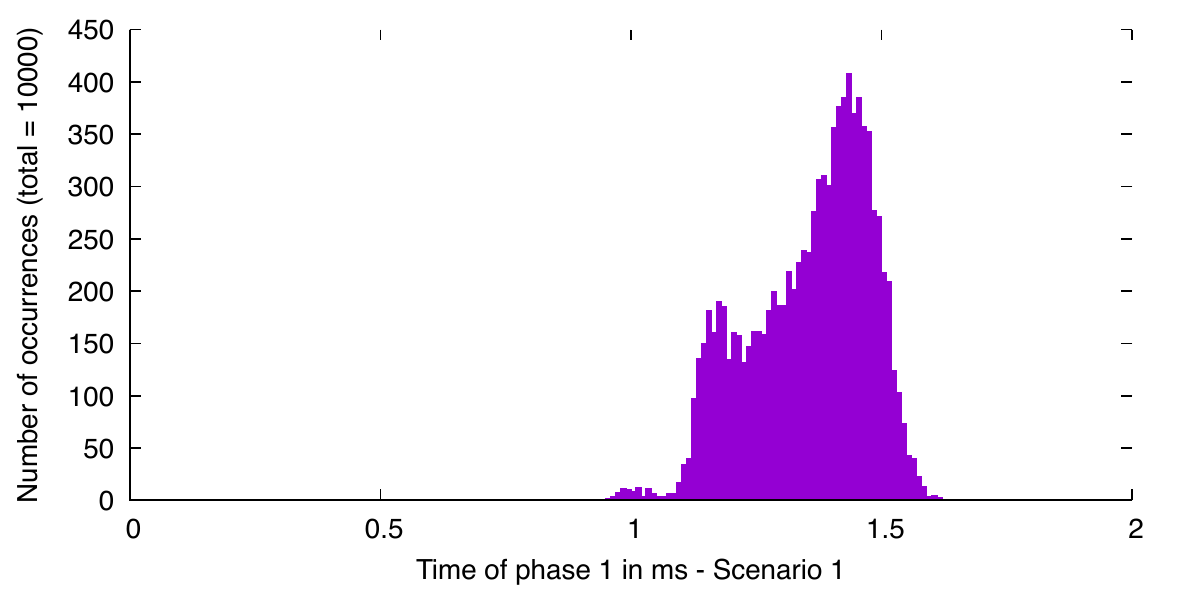}
	\includegraphics[width = 8cm]{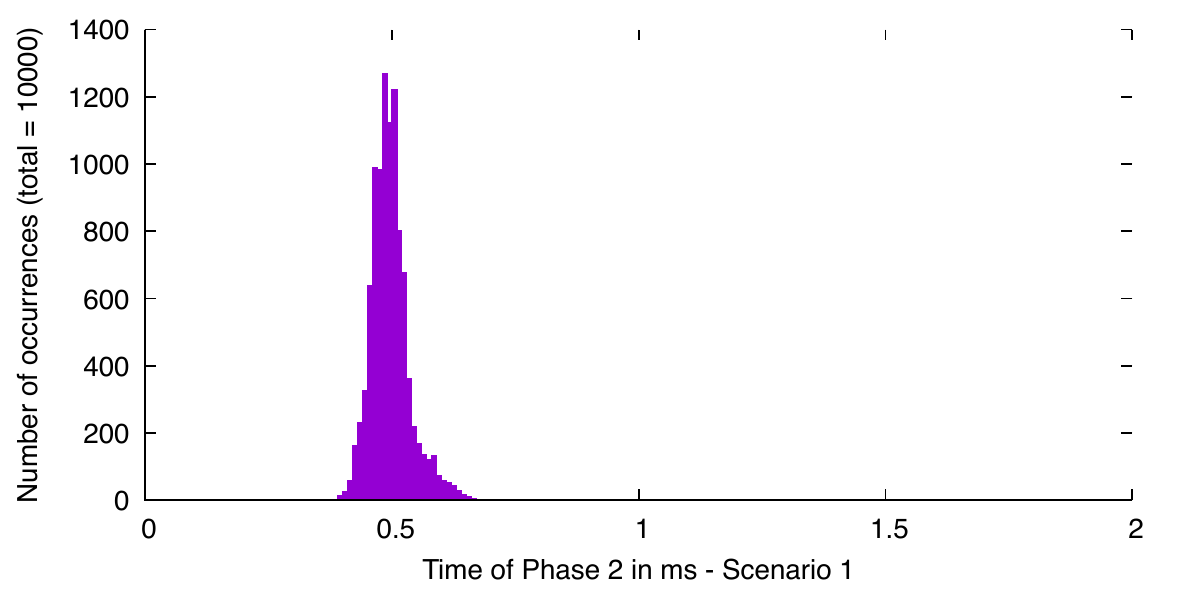}
	\caption{$\TPO$ and $\TPT$ for Scenario $1$ with $10000$ rounds, times are aggregated in intervals of size $0.01ms$.}
	\label{Figure:TimingsScenario1}
\end{figure}

\begin{figure}[t]
	\includegraphics[width = 8cm]{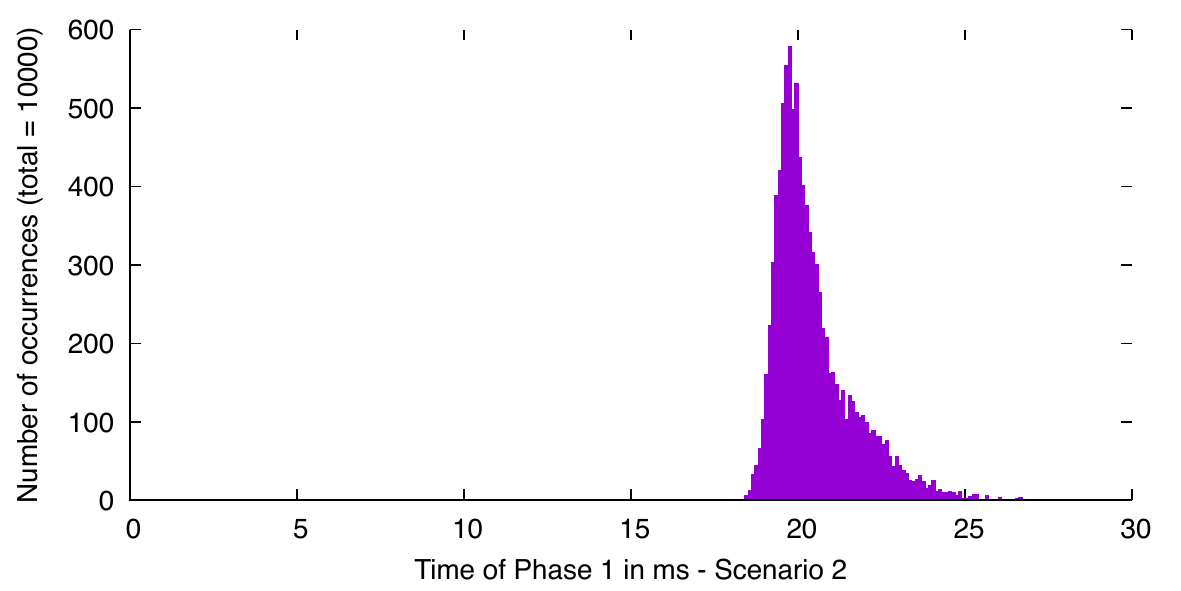}
	\includegraphics[width = 8cm]{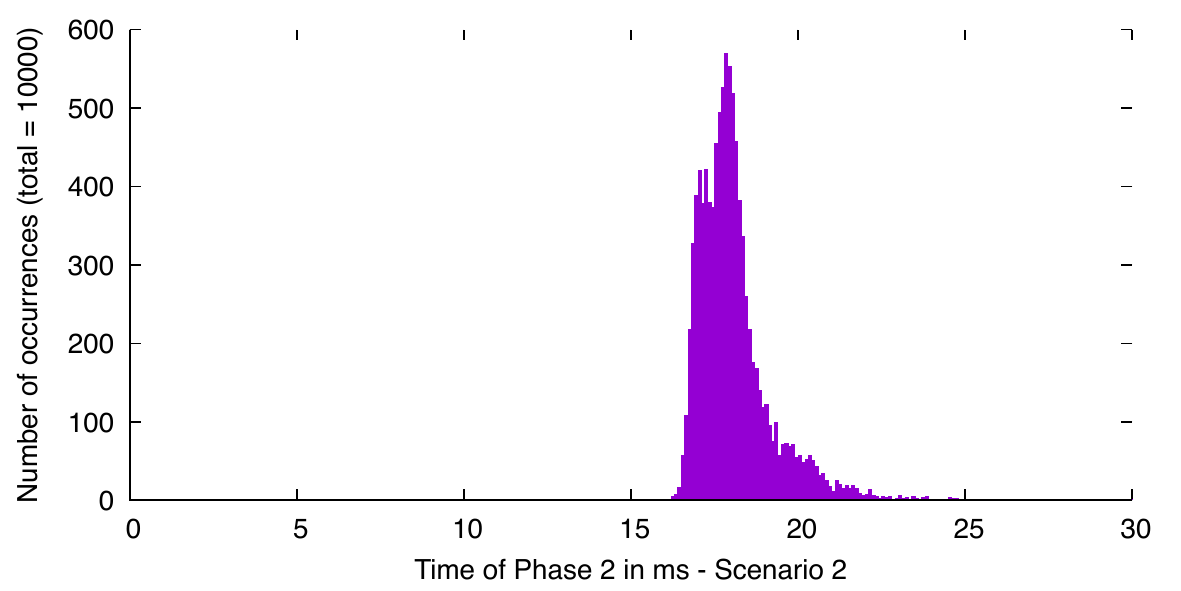}
	\caption{$\TPO$ and $\TPT$ for Scenario $2$ with $10000$ rounds, times are aggregated in intervals of size $0.1ms$.}
	\label{Figure:TimingsScenario2}
\end{figure}

In order to prove soundness, the probability that a cheating prover succeeds in the No case is bounded by $P^*(R,F)$. If we use Equation \ref{Equation:D2} from Appendix~\ref{Appendix:FullProtocol} with $\omega^*(\GRS) \le \frac{2}{3} + 2^{-138}$ and $\lambda = \frac{22}{340}$, we obtain $P^*(R,F) \le 2^{-103}$. For the completeness error, we estimate the probability that the space-time constraint is not satisfied with $p_{\text{loss}} \le \frac{1}{1000}$, which is larger than what we actually observe. With this estimate, if we define $CE(R,F,p_{\text{loss}})$ the probability of failure in the honest case, we obtain from Equation \ref{Equation:D3} that $CE(R,F,p_{\text{loss}}) \le 2^{-102}$. We therefore get the following results, valid for the $2$ scenarios:

\begin{theorem}
	In our experiments, the probability that the verifier rejects an honest run of the protocol is $2^{-102}$ (Completeness error), the soundness is $2^{-103}$ and it is perfect zero-knowledge.
\end{theorem}

We performed the benchmarks with our working laptops and desktops both in Paris and Strasbourg. In this first scenario, we actually connected two desktops from the same local network of our research group. In the second scenario, we performed a communication between our local laptop in Paris and a distant desktop in Strasbourg, which is $400km$ away. We use this setup both for analyzing the running time of our protocol between $\pp_1,\vv_1$ and $\pp_2,\vv_2$ in the $2$ scenarios. While we didn't perform aa fully integrated implementation into a larger cryptographic protocol, the results we obtain show the feasibility and practicality of our relativistic scheme. Also, it is quite flexible on the locations of the verifiers, as shows by our $2$ scenarios and we don't use dedicated hardware for the communication and the computation so one could have even more flexibility with better hardware but again, our goal was to show that this protocol can be implemented without specific hardware. Experimental hardwares are Intel Xeon E5-2650 v3 @2.3~GHz (Haswell) and 
Intel Core i5-6300U CPU @2.4~GHz (Skylake).
Our \texttt{C} code does not use parallelism, and was compiled using the GNU C
Compiler~7.5.0, and the GNU Multiple Precision Arithmetic Library (\url{https://gmplib.org/}) to perform arithmetic operations in
$\mathbb{F}_Q$ (with $Q = 2^{23~209} - 1$ in our example).

We also did experiments for other values of $n$ to show to what extent $n$ can be increased before the space-time constraints are not verified anymore. We present our data with increased $n$ in Appendix \ref{Appendix:OtherValues} and if one requires a larger security, it is always possible to make $\vv_1,\vv_2$ farther away or to use better hardware. 

\textbf{Acknowledgments.} AC and YB were supported by ANR DEREC $<$ANR-16-CE39-0001-01$>$.

{\footnotesize \vspace*{-0.5cm} \bibliography{paper}}
\bibliographystyle{alpha}

\newpage
\begin{appendix} $ \ $ \newpage
	\section{Comparison between the different existing schemes}\label{Appendix:Comparison}
	
	The authors of \cite{ABC+21} use parameters for which they have $100$ bits of security. This means both the underlying $\NP$-instance should require time $2^{100}$ to solve using the best quantum algorithms and the soundness should be $2^{-100}$.  We will use this benchmarking for comparing the different schemes. 
	
	\begin{itemize}
		\item  For \cite{CL17}. The best algorithm for Hamiltonian cycle on a graph $G$ with $n$ vertices runs in time $O(1.657^n)$\cite{Bjo14} so we need $n \ge 138$ in order to achieve $100$ bits of security. The best running time is actually performed by a classical algorithm, we don't know a better quantum algorithm for this problem. This protocol requires to commit each bit of the upper triangle of the adjacency matrix of $G$ using the relativistic $\F_Q$ bit commitment scheme with $Q > 10000n!$ (to have soundness close to $1/2$) so each round of communication requires at least $2*\frac{n(n-1)}{2}\log_2(10000*Q)$ bits of communication. $\frac{n(n-1)}{2}$ is the number of bits in the upper triangle of the adjacency matrix. The factor $2$ comes from the fact that the verifier first sends $\frac{n(n-1)}{2}$ elements of $\FQ$ and then the prover sends back $\frac{n(n-1)}{2}$ elements of $\FQ$. For $n = 138$, this gives a communication of  $1.51 \cdot 10^7 $ bits which is approximately $1.89MBytes$. In order to achieve soundness of $2^{-100}$, we use $100$ rounds since each round has soundness $\frac{1}{2}$.
		\item The protocol of \cite{CMS+20} and its implementation in \cite{ABC+21} uses a graph $G$ with $588$ vertices and $1097$ edges. The communication is essentially sending a edge which is of size less than $2$ bytes. In order to achieve a soundness of $2^{-100}$, the number of repetition they use is $10^6$. These parameters achieve classical security. If we wanted to achieve quantum security, \cite{ABC+21} claims that this would require $(21*|E|)^4*100$ which is more than $10^{19}$, as well as a third prover/verifier pair. 
		\item The best quantum algorithm for random instances of the syndrome decoding problem requires time at least $2^{0.05869n}$ \cite{KT17} for $k \approx 0.4514n$ and $w \approx 0.1268n$.  Our protocol uses $n = 1704, k = 769, w = 216$ in order to have $100$ bits of security. At each round, we have to commit to $3$ strings of $\FQ$ which means the communication at each round is $6\log_2(Q)$. We can prove security of our scheme by taking $Q = 10^{12}n!2^{4n}$, so the communication is $6\log_2(Q) = 136177$ bits which is equal to $17.03KB$. The number of rounds we take is $340$, which allows our protocol to be loss tolerant.
	\end{itemize}

\section{Stern's zero-knowledge protocol for the syndrome decoding problem}\label{Appendix:Stern}
\noindent We now describe Stern's zero-knowledge protocol~\cite{Ste93} for the Syndrome Decoding problem. It uses a commitment scheme that we don't explicit here. 
\cadre{
	\begin{center}
		Stern's single round zero-knowledge protocol for $\SD$.
	\end{center}
	\textbf{Input:} Integers $k,w$, a matrix $\Hm \Unif \zo^{(n-k) \times n}$, a column vector (called the syndrome) $\sv \in \zo^{n-k}$. \\
	\textbf{Auxiliary input:} A column vector $\ev \in \zo^n$ such that $|\ev|_H = w$ and $\Hm  \ev = \sv$.\\
	\textbf{Protocol:}
	\begin{enumerate}\setlength\itemsep{-0.2em}
		\item The prover picks a random permutation $\sigma$ acting on $[n]$ and a random column vector $\tv \in \zo^n$. Let $s' = \Hm  \tv$, $z_1 = (\sigma || \sv')$, $z_2 = \sigma(\tv)$, $z_3 = \sigma(\tv \oplus \ev)$, where permuting a vector means permuting its coordinates. He commits to $z_1,z_2$ and $z_3$ separately and sends these commitments $\comm_1,\comm_2,\comm_3$.
		\item The verifier sends a uniformly random challenge $c \in \{1,2,3\}$.
		\item The prover opens $z_{c'}$ for the two values $c'$ different from  $c$.
		\item The verifier checks the validity of the $2$ commitments and also performs the following checks:
		\vspace*{-0.2cm}
		\begin{itemize}\setlength\itemsep{-0.2em}
			\item if $c = 1$, accept iff. $|z_2 + z_3|_H = w$.
			\item if $c=2$,  accept iff. ${\Hm} \cdot \sigma^{-1}(z_3) = \sv \oplus \sv'.$ 
			\item if $c=3$, accept iff.  ${\Hm} \cdot \sigma^{-1}(z_2) = \sv'.$
		\end{itemize}
	\end{enumerate}
} $ \ $ \\ \\
This protocol was shown to be secure in \cite{Ste93}. We reproduce here the main aspects of this proof. 
\paragraph{Completeness}
The protocol has perfect completeness. Indeed, in the honest case:
\begin{enumerate} \setlength\itemsep{-0.2em}
	\item $|z_2 + z_3|_H= |\sigma(\tv) \oplus \sigma(\tv \oplus \ev)|_H = |\sigma(\ev)|_H = w.$
	\item $\Hm \cdot \sigma^{-1}\left(\sigma(\tv \oplus \ev)\right) = \Hm \tv \oplus \Hm\ev = \sv \oplus \sv'$.
	\item $\Hm \cdot \sigma^{-1}(\sigma(\tv))= \Hm\tv = \sv'.$
\end{enumerate}
\paragraph{Soundness}
We are in the NO case, so there are no vectors $\ev$ such that $|\ev|_H = w$ and $\Hm\ev = \sv$. We will prove a proposition which is closely related to the soundness (more precisely the $3$-special soundness) of the scheme.
\begin{proposition}\label{Proposition:3ss}
	In the NO case, assume the prover manages to successfully answer the $3$ challenges at the same time for the same first message, then he is able to successfully produce $2$ different openings for the same commitment.
\end{proposition}
\begin{proof}
	We will prove this proposition by contradiction. We are in the NO case, and assume the prover successfully answers the $3$ challenges from the same first message $M_1 = \comm_1,\comm_2,\comm_3$. Assume by contradiction that  for each challenge, he uses the same openings for the commitment $y_1,y_2,y_3$, which we call $z_1 = (\sigma || \sv')$, $z_2,z_3$. Let $\ev_0 = \sigma^{-1}(z_2) \oplus \sigma^{-1}(z_3) = \sigma^{-1}(z_2 \oplus z_3)$. Since the prover successfully answers challenge $1$, we have $|z_2 \oplus z_3|_H = w$ hence $|\ev_0|_H = w$. Moreover, since he successfully answers challenges $2$ and $3$, we have
	$$ \Hm\ev_0 = \Hm(\sigma^{-1}(z_2) + \sigma^{-1}(z_3)) = \sv \oplus \sv' \oplus \sv' = \sv.$$
	Since we are in the NO case, such an $\ev_0$ doesn't exist, hence the contradiction.
\end{proof}
This means that a cheating prover can cheat with probability at most $\frac{2}{3}$ unless he is able to break the binding property of the commitment scheme. 
\paragraph{Zero-knowledge}
The protocol is known to be zero-knowledge if the commitment scheme is hiding, we sketch the proof here. The verifier doesn't get any information from the commitments from the hiding property of the commitment scheme. If the verifier sends the challenge $c=1$, he receives $\sigma(\tv)$ and $\sigma(\tv+\ev)$ from which he cannot recover any information because $\sigma$ is unknown. For the challenge $c = 2$, he receives $\sigma,\sv',\sigma(\tv+\ev)$ hence he can recover $\tv + \ev$. However, since $\tv$ is unknown, this looks like a random vector. For the challenge $c = 3$ he receives $\sigma,\sv',\sigma(\tv)$ , which are random elements independent of $\ev$.

\section{Proof of the security of our $1$ round relativistic zero-knowledge protocol}\label{Apendix:SecurityProof}
The goal of this section is to prove the security of our $1$-round relativistic zero-knowledge for $\NP$ presented in Figure \ref{Figure:1round}.

\paragraph{Completeness.} Completeness follows directly from the completeness of Stern's single round, and from the fact that the $\F_Q$ string commitment has perfect completeness. \\

\paragraph{Zero-knowledge.} Stern's signature scheme is perfectly zero-knowledge if the commitment scheme is perfectly hiding, which is the case of the $\F_Q$ string commitment. Therefore, the protocol is perfectly zero-knowledge. This zero-knowledge property is also preserved when the scheme is repeated sequentially. \\
\paragraph{Soundness.}  Proving soundness against quantum adversaries, is the main technical challenge of this work. We start from a NO instance meaning that there is no solution to the $\SD$ problem and we consider an all powerful cheating prover that wants to convince the verifier such a solution exists. 

In all generality, $\pp_1$ and $\pp_2$ can share an entangled state $\ket{\Phi}$. Then $\pp_1$ receives $B = b_1,b_2,b_3$ from $\vv_1$ and $\pp_2$ receives $c \in \{1,2,3\}$ from $\vv_2$. They output respectively $Y = y_1,y_2,y_3$ and $ZA = \{(z_{c'},a_{c'})\}_{c' \neq c}$. The verifiers then come together and perform the verification step. We assume the timing constraints are verified which ensures that $\pp_1$  has no information about $chall$ before sending his message $Y$ and $P_2$ has no information about $B_1,B_2,B_3$ before sending his message $ZA$.

Any strategy from the provers can be directly related to the strategy for the following $2$-player game  where the $2$ provers play the role of these $2$ players and the verifiers play the role of the referee ({\ie} they send the random questions to the provers and check the validity of their outputs). 
This game, that we call $\GRS$, is defined as follows: \\ \\
\cadre{
	\begin{center} $2$-player game $\GRS$ \end{center}
	\begin{itemize}
		\item Alice receives $B = b_1,b_2,b_3 \in_R \FQ$. Bob receives $c \in_R {1,2,3}$.
		\item Alice outputs $Y = y_1,y_2,y_3 \in \FQ$ and Bob outputs $AZ = \{(a_{c'},z_{c'}\}_{c' \neq c}$ where each $a_i\in \FQ$.
		\item 
		Each $z_i$ is an element of $\FQ$ but we interpret $z_1$  as an element $(\sigma,\sv')$ of  $S_1 = P_n \times \zo^n$. Similarly, we interpret $z_2,z_3$ as vectors in $\zo^n$ using a mapping that we detail after the description of the game. If this mapping fails, the game is lost. Otherwise, we first check the constraint $y_{c'} = a_{c'} + z_{c'}*b_{c'}$ for $c' \neq c$.	 We then check:
		\begin{enumerate}
			\item if $c = 1$, we also require $|z_2 + z_3|_H = w$.
			\item if $c=2$,  we also require ${\Hm} \cdot \sigma^{-1}(z_3) = \sv \oplus \sv'.$ 
			\item if $c=3$, we also require ${\Hm} \cdot \sigma^{-1}(z_2) = \sv'.$
		\end{enumerate}
	\end{itemize}
}

The $``+"$ and $``*"$ operations we use are the one in $\FQ = \{\overline{0},\dots,\overline{Q-1}\}$ where $\overline{i}$ is the $i^{\mbox{th}}$ element of $\FQ$. We do our mapping as follows: If $z_1 \in \{\overline{0},\dots,\overline{S_1 - 1}\}$, then we map $z_1$ to the $z_1^{\mbox{th}}$ element of $S_1$, otherwise, we say that the mapping fails. We do the same thing for $z_2,z_3$, if they are in $\{\overline{0},\dots,\overline{2^n-1}\}$ then we can map them to binary vectors, otherwise, we say that the mapping fails. In order for this mapping to be well defined, we must take $Q$ large enough, more precisely $Q \ge |S_1|$ and $Q \ge 2^n$.

What is the optimal cheating strategy we can expect for this game? There are some strategies that can win Stern's single round zero-knowledge protocol with probability $\frac{2}{3}$ \cite{Ste93} from which we can directly derive strategies for this game with $2$ classical players that win wp. $\frac{2}{3}$. This means we have strategies for which the players win for $2$ possible challenges for Bob but not for the third one. 

So what we want to show is that there is no strategy for which the players will win for the $3$ challenges received by Bob. What we know from Proposition \ref{Proposition:3ss} is that they can't give answers for the $3$ challenges simultaneously but this doesn't mean they can't answer each challenge separately. This behavior can appear when we consider entangled strategies. For example in magic square game \cite{Mer90,Per90}, Alice and Bob can't answer all questions at the same time --- because all the constraints of the magic square game lead to a contradiction. However, the entangled value of the game is still $1$.

Our main contribution is to bound the entangled value of $\GRS$. As $Q$ increases, the entangled value converges to $\frac{2}{3}$ which is optimal.  We prove the following theorem

\begin{theorem}\label{Theorem:GRS}
$$	\omega^*(\GRS) \le \frac{2}{3} + \left(\frac{n!2^{4n}}{Q}\right)^{1/4}.$$
\end{theorem}

From this theorem, our result on the soundness of our $1$-round relativistic scheme will follow immediately. In the next section, we will prove the bound in $\omega^*(\GRS)$.

\subsection{Bounding the value of the game}
We can now prove our lower bound on the entangled value of $\GRS$.
\begin{proof}[Proof of the lower bound of the game]
	Let any $\delta > 0$ and consider a finite dimensional projective strategy for Alice and Bob that wins  the game $\GRS$ wp. $\omega^*(\GRS) - \delta$.
	
	Let $P^{B} = \{P^{B}_{Y}\}$ and $Q^c = \{Q^c_{A}\}$ be respectively Alice's and Bob's projective measurements for their respective inputs $B = b_1,b_2,b_3$ and $c$. Alice's output is $Y = y_1,y_2,y_3$ and Bob's output $AZ$ corresponds to the $2$ pairs $(a_{c'},z_{c'})$ for $c' \neq c$, starting with the one with smallest index. Let $\ket{\psi}$ be the quantum state they share. 

Fix an input/output pair $BY$ for Alice and let $\sigma^{BY}$ be the state held by Bob, conditioned on this pair. For each $c \in \{1,2,3\}$, let $W_c = \{AZ: V(Y,AZ,B,c)= 1\}$ be the set of winning outputs for Bob where $V(Y,AZ,B,c)$ is the function that outputs 1 if the game is won on inputs/outputs $(B, c)/(Y, AZ)$, and outputs 0 otherwise.

A necessary condition of validity (for a fixed $BY$), is that $y_i = a_i + b_{i}*z_i$ for the $z_i$ revealed so for each $z_i$, there is a unique valid $a_i$ which is $y_i - b_i*z_i$.
For $c = 1$, Bob outputs $a_2,z_2,a_3,z_3$. A necessary condition is that $z_2,z_3$ can each be mapped into the set of binary vectors. Since there is a $1$ to $1$ correspondence between $z_i$ and $a_i$, we have $|W_1| \le 2^{2n}$. For $c = 2$, Bob outputs $a_1,z_1,a_3,z_3$. A necessary condition is that $z_1$ can be mapped to an element  $(\pi,\sv')$ where $\pi$ is a permutation on $[n]$ and $\sv' \in \zo^n$ and $z_3$ has to be mapped to an binary vector, which implies $|W_2| \le 2^{n}n!\cdot 2^{n} = 2^{2n}n!$. A similar reasoning for $c = 3$ gives $|W_3| \le 2^{2n}n!$.

For each $c \in \{1,2,3\}$, let $Q^c_W = \sum_{AZ \in W_c} Q^c_{AZ}$ the projector on the winning outputs for Bob on input $c$ (for the fixed input/output pair $BX$ of Alice).
Let $V^{BX}$ be the probability that Alice and Bob win the game for this input. We have 
$$ V^{BX} =  \frac{1}{3}\sum_{c \in \{1,2,3\}} tr(Q^c_W \sigma^{BX}).$$
and also 
$\E_{BX}[V^{BX}] = \omega^*(\GRS) - \delta$. We now consider the following quantum strategy for Bob that will make him succeed on the $3$ challenges: wp. $\frac{1}{2}$, run $Q^1$ to get output $AZ_1$, then on the resulting state, run $Q^2$ to get output $AZ_2$ and on the resulting state, run $Q^3$ to get output $AZ_3$. Wp. $\frac{1}{2}$, do the same thing but swap the order ot $Q^1$ and $Q^2$. Let $E^{BX}$ be the probability of success of this strategy. We can write

\begin{multline*}E^{BX} \eqdef 
\frac{1}{2}\Big(\sum_{\substack{AZ_1 \in W_1 \\ AZ_2\in W_2 \\ AZ_3 \in W_3}}
tr\left(Q^3_{AZ_3}Q^2_{AZ_2}Q^1_{AZ_1} \sigma^{BX}Q^1_{AZ_1}Q^2_{AZ_2}\right) + \\
tr\left(Q^3_{AZ_3}Q^1_{AZ_1}Q^2_{AZ_2} \sigma^{BX}Q^2_{AZ_2}Q^1_{AZ_1}\right)\Big.
\end{multline*}

Notice that for any $3$ projectors, $P_1,P_2,P_3$, we have 
$tr(P_3P_2P_1\sigma P_1P_2P_3) = tr(P_3^2P_2P_1\sigma P_1P_2) = tr(P_3P_2P_1\sigma P_1P_2)$ hence the expression $E^{BX}$.

In order to conclude, we use the $2$ following equations, which will be proven in upcoming proposition. The first one 

\begin{align}\label{Eq:C3}\E_{BX}\left[E^{BX}\right] = \frac{1}{Q},
\end{align}

claims that our strategy will succeed in answering valid outputs $A_1,A_2,A_3$ for the $3$ challenges wp. at most $\frac{1}{Q}$. In high level, this is a direct consequence of Proposition \ref{Proposition:3ss} and of the binding property of the $\FQ$-relativistic commitment scheme, but we reprove this claim from scratch. We then relate $E_{BX}$ and $V_{BX}$ using a generic proposition on projectors 
\begin{align}\label{Eq:C4}
	E^{BX} \ge \frac{9 (V^{BX} - \frac{2}{3})^4}{2|W_1||W_2|}.
\end{align}
Proving this inequality is actually where we had most of the technical difficulty. In order to prove this statement, we generalized the approach of~\cite{CL17} to $3$ measurements and showed that if you can win for the $3$ challenges at the same time wp. at most $\frac{1}{Q}$ then a quantum adversary can win at most wp. $\frac{2}{3} + \eps$ where $\eps$ is vanishingly small for $Q$ large enough.

We first conclude and then go on proving Equations \ref{Eq:C3} and \ref{Eq:C4}. To conclude our proof, we have from Equation \ref{Eq:C4} that 
$$ E^{BX} \ge  \frac{9 (V^{BX} - \frac{2}{3})^4}{2\cdot2^{4n}n!}.$$ 

By taking the expectation on each side, we obtain	 
\begin{align*}
	\frac{1}{Q} &= \E_{BX}[E^{BX}] \ge \E_{BX}[\frac{9 (V^{BX} - \frac{2}{3})^4}{2\cdot2^{4n}n!}] \ge \frac{9 \left(\omega^*(\GRS) - \delta - \frac{2}{3}\right)^4}{2\cdot2^{4n}n!}
\end{align*}
where we used the convexity of the function $x \mapsto x^4 - \frac{2}{3}$. Since this holds for any $\delta > 0$, we take $\delta \rightarrow 0$ and have 
$$ \frac{1}{Q} \ge \frac{9 \left(\omega^*(\GRS) - \frac{2}{3}\right)^4}{2\cdot2^{4n}n!}.$$
\end{proof}

We now prove our two equations. We first prove  the following

\begin{lemma}[Equation \ref{Eq:C3}]
	In the NO case, the probability that Bob successfully outputs $3$ valid couples $(AZ_1,AZ_2,AZ_3)$, on average on $BY$ is at most $\frac{1}{Q}$.
\end{lemma}
\begin{proof}
	Fix an input/output $BY$ with $B = b_1,b_2,b_3$ and $Y = y_1,y_2,y_3$. Assume by contradiction that Bob can output $3$ valid couples $AZ_1,AZ_2,AZ_3$ with $AZ_c = (a^c_{c'_1},z^c_{c'_1}),(a^c_{c'_2},z^c_{c'_2})$ for the $2$ different values $c'_1,c'_2 \neq c$. 
	Assume by contradiction that $z_{1}^2 = z_1^3 \eqdef z_1$,  $z_{2}^1 = z_2^3 \eqdef z_2$,  $z_{3}^1 = z_3^2 \eqdef z_3$. We map $z_1$ to a pair $(\sigma,\sv')$, $z_2,z_3$ to vectors $ \in \zo^n$. Passing the $3$ winning conditions implies that 
	$$\Hm \sigma^{-1}\left(z_2 + z_3\right) = \sv \quad \textrm{and} \quad |\sigma^{-1}\left(z_2 + z_3\right)| = w.$$
	This implies that $\sigma^{-1}\left(z_2 + z_3\right) $ is a solution to the syndrome decoding problem but since we are in the NO case, such a solution doesn't exist hence the contradiction. 
	
	This means there exists $c' \in \{1,2,3\}$ st. $z_{c_1}^{c'} \neq z_{c_2}^{c'}$ for the two values $c_1,c_2 \neq c'$. Because these are valid answers, this means we have 
	\begin{align*}
	z^{c_1}_{c'} * b_{c'} & = y_{c'} - a^{c_1}_{c'} \\
	z^{c_2}_{c'} * b_{c'} & = y_{c'} - a^{c_2}_{c'} 
	\end{align*}
	From which we get 
	$$ b_{c'} = (a_{c_2}^{c'} - a_{c_1}^{c'})\slash (z_{c_1}^{c'}  - z_{c_2}^{c'}).$$
	where $\slash$ is the division in $\F_Q$. From there, this means Bob can guess $b_{c'}$	but from non-signaling, Bob should have no information about $b_{c'}$. Moreover, notice that Bob knows which $c'$ to take, it is the index where $z^{c_1}_{c'} \neq z^{c_2}_{c'}$. Since $b_{c'}$ is a random element from $\FQ$, we conclude that Bob can guess this value wp. $\frac{1}{Q}$ which concludes the proof. 
\end{proof}

The next section is devoted to the proof of the second equation.

\subsection{Proof of Equation \ref{Eq:C4}}

We prove the following
\begin{proposition}\label{Proposition:Hard}
	Consider $3$ projectors $P_1,P_2,P_3$ such that for each $i \in \{1,2,3\}$, we can write $P_i = \sum_{s_i = 1}^{S_i} P^s_i$ where for each $i$, the $\{P^s_i\}$ are orthogonal projectors meaning that $P^s_iP^{s'}_i = \delta_{s,s'}P^s_i$. Let $\sigma$ be any quantum state.
	Let $V \eqdef \sum_i tr(P_i \sigma)$ and 
	\begin{multline*}
		E \eqdef \frac{1}{2}\Big(\sum_{s_3 = 1}^{S_3} \sum_{s_2 = 1}^{S_2} \sum_{s_1 = 1}^{S_1} tr\left(P_3^{s_3}P_2^{s_2}P_1^{s_1} \sigma  \left(P_1^{s_1}\right) \left(P_2^{s_2}\right)\right) + \\ \sum_{s_3 = 1}^{S_3} \sum_{s_2 = 1}^{S_2} \sum_{s_1 = 1}^{S_1} tr\left(P_3^{s_3}P_1^{s_1}P_2^{s_2} \sigma  \left(P_2^{s_2}\right) \left(P_1^{s_1}\right) \right)\Big).
	\end{multline*}
	We have 
	$E \ge \frac{9\left(V - \frac{2}{3}\right)^4}{2S_1S_2}$.
\end{proposition}

In order to prove our proposition, we first need the following trigonometric lemma.

\begin{lemma}\label{lemma:TrigoIneq}
	$\forall \alpha_1,\alpha_2 \in [0,\pi/2[$ st. $\cos^2(\alpha_1) + \cos^2(\alpha_2) > 1$, if we define $c_1 = \cos(\alpha_1),c_2 = \cos(\alpha_2), s_1 = \sin(\alpha_1), s_2 = \sin(\alpha_2)$, we have 
	$$ \min_{y \in [-s_1,s_1]} \left\{\frac{c^2_2\left(c_1c_2+ ys_2\right)^2}{\left(c_1c_2 + ys_2\right)^2 + s_1^2 - y^2}\right\} \ge c^2_1 + c^2_1- 1.$$
\end{lemma}
\begin{proof}
	Let $t_1 = \tan(\alpha_1)$ and $t_2 = \tan(\alpha_2)$. Since $c_1^2 + c_2^2 > 1$, we have $\alpha_1 + \alpha_2 < \pi/2$ and $t_1t_2 < 1$.
	Let also 
	$$ T(y) \eqdef (c_1c_2 + ys_2)^2  \quad ; \quad U(y) \eqdef  s_1^2 - y^2.$$
	Our goal is hence to minimize the function 
	$$ f(y) \eqdef \frac{c_2^2 T(y)}{T(y) + U(y)} \quad \textrm{for } y \in [-s_1,s_1].$$
	We write 
	\begin{align*}
		f'(y) & = \dfrac{c^2_2T'(y)\left(T(y) + U(y)\right) - c_2^2T(y)\left(T'(y) + U'(y)\right)}{\left(T(y) + U(y)\right)^2} \\ & = \dfrac{c_2^2  \left(T'(y)U(y) - T(y)U'(y)\right)}{\left(T(y) + U(y)\right)^2}
	\end{align*}
	and 
	\begin{align*}
		Z_0 & = T'(y)U(y) - T(y)U'(y) \\ & = (2ys_2^2 + 2c_1c_2s_2)(s_1^2 - y^2) - (y^2s_2^2 + 2yc_1c_2s_2 + c_1^2c_2^2)(-2y) \\
		& = 2y^2(c_1c_2s_2) + 2y(s_1^2s_2^2 + c_1^2c_2^2) + 2c_1c_2s_1^2s_2 \\
		& = 2c_1c_2s_2\left(y + \dfrac{s_1^2s_2}{c_1c_2}\right)\left(y + \dfrac{c_1c_2}{s_2}\right)
	\end{align*}
	
	This implies the equation $f'(y) = 0$ has $2$ solutions $y_0 = -\frac{s_1^2s_2}{c_1c_2} = -s_1t_1t_2$ and $y_1 = -\frac{c_1c_2}{s_2} = -\frac{s_1}{t_1t_2}$. Notice that only $y_0$ lies in the interval $[-s_1,s_1]$ since $t_1t_2 < 1$.
	
We now write 
	$f(-s_1) = f(s_1) = c_2^2$ and 
	\begin{align*}
		T(y_0) & = \dfrac{1}{c_1^2c_2^2} \left(c_1^2c_2^2 - s_1^2s_2^2\right)^2 \\
		U(y_0) & = \dfrac{s_1^2}{c_1^2c_2^2}\left(c_1^2c_2^2 - s_1^2s_2^2\right) 
	\end{align*}
	which implies 
	\begin{align*}
		f(y_0) & = \frac{c_2^2\left(c_1^2c_2^2 - s_1^2s_2^2\right)}{\left(c_1^2c_2^2 - s_1^2s_2^2\right) + s_1^2} \\
		& = c_2^2 - s_1^2 & \textrm{using } \left(c_1^2c_2^2 - s_1^2s_2^2\right) = c_2^2 - s_1^2 \\
		& = c_1^2 + c_2^2 - 1
	\end{align*}
	In order to conclude, we write 
	$$f(-s_1) = c_2^2 \ ; \  f(s_1)  = c_2^2 \ ; \ f(y_0) = c_1^2 + c_2^2 - 1 \le f(s_1).$$
	Since $y_0$ is the unique point in $[-s_1,s_1]$ such that $f'(y_0) = 0$ we can conclude that 
	$$\min_{y \in [-s_1,s_1]}\{f(y)\} = f(y_0) = c_1^2 + c_2^2 - 1.$$
\end{proof}

We now prove our proposition for the special case $S_1,S_2 = 1$ and for a pure state $\ket{\Omega}$ instead of $\sigma$.

\begin{proposition}\label{Proposition:S=1}
	Let $\ket{\Omega}$ be a quantum state. Let $P_1,P_2,P_3$ be projectors. Let $V = \frac{1}{3} \norm{P_i \ket{\Omega}}^2 = \frac{2}{3} + \eps$ with $\eps \ge 0$ and $E = \frac{1}{2}\left(\norm{P_3P_2P_1 \ket{\Omega}}^2 + \norm{P_3P_1P_2 \ket{\Omega}}^	2\right)$. We have $ E \ge \frac{9\eps^4}{2}.$
\end{proposition}
\begin{proof}
	Let $\ket{\phi_i} = \frac{P_i\ket{\Omega}}{\norm{P_i \ket{\Omega}}}$. 
	We write 
	\begin{align}
	\ket{\Omega} = \cos(\alpha_i) \ket{\phi_i} + \sin(\alpha_i) \ket{\phi_i^\bot} \ \textrm{ for } i \in \{1,2,3\}.
	\end{align} This means 
	$$V = \frac{1}{3} \left(\cos^2(\alpha_1) + \cos^2(\alpha_2) + \cos^2(\alpha_3)\right).$$
	We write 
	\begin{align}
	\ket{\phi_2} & = \cos(\alpha_2)\ket{\Omega} + \sin(\alpha_2)\ket{B}
	\end{align} for some pure state $\ket{B} \bot \ket{\Omega}$ and 
	\begin{align}\ket{\phi_1} = \cos(\alpha_1)\ket{\Omega} + x \ket{A} + y \ket{B}
	\end{align} for some pure state $\ket{A} \bot \ket{B}$ and $\ket{A} \bot \ket{\Omega}$. This means we have 
	\begin{align}\label{Eq:C6}
		\cos^2(\alpha_1) + |x|^2 + |y|^2 = 1.
	\end{align}
	We also write 
	\begin{align}P_2 = \kb{\phi_2} + P'_2 \ \textrm{ with } P'_2 \ket{\phi_2} = \zv.
	\end{align} We have 
	\begin{align}
		\ket{W} = P_2P_1(\ket{\Omega}) & = \cos(\alpha_1)(P_2 \ket{\phi_1}) \nonumber \\
		& = \cos(\alpha_1)\braket{\phi_1}{\phi_2}\ket{\phi_2} + \cos(\alpha_1)P'_2\ket{\phi_1}
	\end{align}
	
	Notice that $P_2 \ket{\Omega} = \cos(\alpha_2)\ket{\phi_2} = \kb{\phi_2} \cdot \ket{\Omega}$ hence $P'_2 \ket{\Omega} = \zv$. This implies that $P'_2 \ket{B} = \zv$ and  
	\begin{align}
	P'_2 \ket{\phi_1} = P'_2 (x\ket{A}) = z \ket{A'}
	\end{align}
	for some $\ket{A'}$ orthogonal to $\ket{\Omega}$ and $\ket{B}$ and $|z| \le |x|$. So we rewrite
	\begin{align}
		\frac{1}{\cos(\alpha_1)}\ket{W} & = \braket{\phi_1}{\phi_2}\ket{\phi_2} + P'_2\ket{\phi_1} \nonumber \\
		& = \left(\cos{\alpha_1}\cos(\alpha_2) + \sin(\alpha_2)y\right) \ket{\phi_2} + z \ket{A'} 
	\end{align}
	Let $u = \left(\cos{\alpha_1}\cos(\alpha_2) + \sin(\alpha_2)y\right)$ so that 
	\begin{align}\label{Eq:C10}
	\ket{W} = \cos(\alpha_1)u\ket{\phi_2} + \cos(\alpha_1)z\ket{A'}
	\end{align} The norm of $\ket{W}$ is therefore
	$$ \norm{\ket{W}} = |\cos(\alpha_1)| \sqrt{|u|^2 + |z|^2}.$$
	Let $\ket{\widetilde{W}} = \ket{W}/\norm{\ket{W}}$. From Equation \ref{Eq:C10}, we have 
	\begin{align*}|\braket{\Omega}{\widetilde{W}}|^2  & = (\frac{|\cos(\alpha_2)\cos(\alpha_1)u|}{\norm{\ket{W}}})^2 = \frac{\cos^2(\alpha_2)u^2}{u^2 + z^2} \\ & \ge \frac{\cos^2(\alpha_2)u^2}{u^2 + (1 - \cos^2(\alpha_1) - y^2)} .\end{align*}
	
	Using lemma \ref{lemma:TrigoIneq}, we obtain  $|\braket{\Omega}{\widetilde{W}}|^2  \ge \cos^2(\alpha_1) + \cos^2(\alpha_2) - 1$. We hence write $|\braket{\Omega}{\widetilde{W}}|^2  = \cos^2(\beta)$ for some $\beta \le \alpha_1 + \alpha_2$. In order to conclude, we define $Angle(\ket{\psi},\ket{\phi}) = \arccos(|\braket{\psi}{\phi}|)$. The angle function is a distance measure \cite{NC00}. We will use several times the trigonometric inequality $\cos(\rho + \theta) \ge \cos^2(\rho) + \cos^2(\theta) - 1$ for any $\rho,\theta$ with $\rho + \theta \le \pi/2$  and we can hence write 
	\begin{align*}
		\norm{P_3 \ket{\widetilde{W}}}^2 & \ge |\braket{\phi_3}{\widetilde{W}}|^2 \\
		& = \cos^2\left(Angle\left(\ket{\phi_3},\ket{\widetilde{W}}\right)\right) \\
		& \ge \cos^2\left(Angle(\ket{\phi_3},\ket{\Omega}) + Angle(\ket{\Omega},\ket{\widetilde{W}})\right)  \\
		& \ge \cos^2(\alpha_3 + \beta) \\
		& \ge \left(\cos^2(\alpha_3) + \cos^2(\beta) - 1\right)^2 \\
		& = \left(\cos^2(\alpha_1) + \cos^2(\alpha_2) + \cos^2(\alpha_3) - 2\right)^2 = \eps^2
	\end{align*}
	From there, we can conclude 
	\begin{align*}\norm{P_3P_2P_1 \ket{\Omega}}^2 & = \norm{P_3 \ket{W}}^2 = \norm{P_3 \ket{\widetilde{W}}}^2 \norm{W}^2 \\& \ge \eps^ 2 \norm{P_2P_1\ket{\Omega}}^2
	\end{align*}
	In order to conclude, we use $ |y| \le sin(\alpha_1)$ (Equation \ref{Eq:C6}) which gives $|u| \ge \cos(\alpha_1 + \alpha_2)$, and 
	$$\norm{P_2P_1\ket{\Omega}}^2 = \norm{W}^2 \ge \cos^2(\alpha_1)\cos^2(\alpha_1 + \alpha_2).$$
	which gives 
	\begin{align}
	\norm{P_3P_2P_1 \ket{\Omega}}^2 \ge \eps^2 \cos^2(\alpha_1)\cos^2(\alpha_1 + \alpha_2)
	\end{align}
	Similarly, we have
	\begin{align}\norm{P_3P_1P_2 \ket{\Omega}}^2 & =  \eps^2\norm{P_1P_2 \ket{\Omega}}^2 \norm{P_3 \ket{\widetilde{W}}}^2 \nonumber \\& \ge \eps^ 2 \cos^2(\alpha_2)\cos^2(\alpha_1 + \alpha_2) \end{align}
	\begin{multline*}
		\frac{1}{2} \left(\norm{P_3P_2P_1 \ket{\Omega}}^2 + \norm{P_3P_1P_2 \ket{\Omega}}^2\right) \ge \\
		\frac{\eps^2}{2}(\cos^2(\alpha_1) + \cos^2(\alpha_2))cos^2(\alpha_1 + \alpha_2)
	\end{multline*}
Now,  since $cos^2(\alpha_1) + \cos^2(\alpha_2) \ge 1 + 3 \eps$ (from $V \ge \frac{2}{3} + \eps$), we have $\cos^2(\alpha_1) + \cos^2(\alpha_2) \ge 1$ and $\cos^2(\alpha_1 + \alpha_2) \ge \left(\cos^2(\alpha_1) + \cos^2(\alpha_2) - 1\right)^2 = (3\eps)^2$, from which we conclude 
	$$\frac{1}{2} \left(\norm{P_3P_2P_1 \ket{\Omega}}^2 + \norm{P_3P_1P_2 \ket{\Omega}}^2\right) \ge \frac{9\eps^4}{2}.$$
\end{proof}

We proved our main proposition for $S_1,S_2 = 1$. From there, we can directly go to the general case in similar way than in \cite{CL17}.We need the following statement 

\begin{lemma}[Proposition $4$ from \cite{CL17}]\label{Lemma:FromCL17}
	Let a projector $P = \sum_{i = 1}^m$ where $\{P_i\}_{i \in [m]}$ are orthogonal projectors. For any pure state $\ket{\psi}$, we have 
	$$ \sum_{i = 1}^m P_i \kb{\psi} P_i \ge \frac{1}{m} P \kb{\psi} P.$$
\end{lemma}

With this lemma, we can prove our main proposition.

\begin{proof}[Proof of Proposition \ref{Proposition:Hard}]
	First, in order to use Proposition \ref{Proposition:S=1}, we need to work on pure states similarly as in \cite{CL17}. Assume $\sigma$ is a quantum  mixed state in some Hilbert space $\bb$ and the projectors $P_i^s$ act on $\bb$.  
	We add an extra Hilbert space
	$\mathcal{E}$. We consider a purification $\ket{\Omega}$ of $\sigma$ in $\bb \mathcal{E}$ and define $\widetilde{P}_i = P_i \otimes I_E$,  $\widetilde{P}_i ^s= P_i^s \otimes I_E$, $\widetilde{V} = \sum_i tr(\widetilde{P}_i\ket{\Omega})$ and 
	$$ \widetilde{E} = \frac{1}{2}\left(\widetilde{E}_1 + \widetilde{E}_2\right)$$
	with 
	\begin{align*}
		\widetilde{E}_1 & \eqdef \Big(\sum_{s_3 = 1}^{S_3} \sum_{s_2 = 1}^{S_2} \sum_{s_1 = 1}^{S_1}  tr\left(\widetilde{P}_3^{s_3}\widetilde{P}_2^{s_2}\widetilde{P}_1^{s_1} \kb{\Omega}  \left(\widetilde{P}_1^{s_1}\right) \left(\widetilde{P}_2^{s_2}\right)\right)\\
		\widetilde{E}_2 & = \sum_{s_3 = 1}^{S_3} \sum_{s_2 = 1}^{S_2} \sum_{s_1 = 1}^{S_1} tr\left(\widetilde{P}_3^{s_3}\widetilde{P}_1^{s_1}\widetilde{P}_2^{s_2} \kb{\Omega}  \left(\widetilde{P}_2^{s_2}\right) \left(\widetilde{P}_1^{s_1}\right) \right)\Big).
	\end{align*}
	One can easily check that $\widetilde{V} = V$ and $\widetilde{E} = E$.
	Now, using Lemma \ref{Lemma:FromCL17} twice, we have
	\begin{align*}
		\widetilde{E}_1 & \ge \frac{1}{S_1} \sum_{s_3 = 1}^{S_3} \sum_{s_2 = 1}^{S_2}  tr\left(\widetilde{P}_3^{s_3}\widetilde{P}_2^{s_2}\widetilde{P}_1 \kb{\Omega}  \widetilde{P}_1  \left(\widetilde{P}_2^{s_2}\right)\right) \\
		& \ge \frac{1}{S_1S_2}\sum_{s_3 = 1}^{S_3}   tr\left(\widetilde{P}_3^{s_3}\widetilde{P}_2 \widetilde{P}_1 \kb{\Omega}  \widetilde{P}_1  \widetilde{P}_2 \right) \\
		& = \frac{1}{S_1S_2} tr(\widetilde{P}_3\widetilde{P}_2\widetilde{P}_1\kb{\Omega} \widetilde{P}_1\widetilde{P}_2)
	\end{align*}
	Similarly, we can prove 
	$$ \widetilde{E}_2 \ge \frac{1}{S_1S_2} tr(\widetilde{P}_3\widetilde{P}_1\widetilde{P}_2\kb{\Omega} \widetilde{P}_2\widetilde{P}_1).$$
	In order to conclude, we  use Proposition \ref{Proposition:S=1} which directly gives us the desired result. 
\end{proof}

\section{Security proof of the full scheme and loss-tolerance}\label{Appendix:FullProtocol}
\subsection{Analysis of loss tolerance}
We consider the $R$ round protocol where we allow at most $F$ aborts and let $\lambda = \frac{F}{R}$. An adversary can of course use this allowed number of losses to cheat in the protocol. We consider cheating provers and assume here that the hardware is perfect, which can only help the cheating provers. At the first round, the provers can perform a strategy that aborts with probability $\lambda^*$ and for which they win the game with probability $p^*$ conditioned on not aborting. The probability $P^*(R,F)$ of cheating is therefore 
$$ P^*(R,F) = \lambda^* P_1+ (1-\lambda^*)p^* P_2.$$
where $P_1$ is the probability that the provers win on the $R-1$ remaining rounds and they have $F-1$ aborts left, and $P_2$ is the probability that the provers win on the $R-1$ remaining rounds and they have $F$ aborts left. While computing $P_1$ and $P_2$ is hard, notice that $P_2 \ge P_1$. Moreover, $(1- \lambda^*)p^* \le \omega^*(\GRS)$ and $p^* \le 1$, so $(1-\lambda^*)p^* \le \min\{\omega^*(\GRS),1-\lambda^*\}$. We now distinguish $2$ cases
\begin{itemize}
	\item If $\lambda^* \in [0,1 - \omega^*(\GRS)]$, $P^* \le \lambda^* P_1 + \omega^*(\GRS) P_2$ and this right hand side is increasing in $\lambda^*$.
	\item If $\lambda^* \in [1 - \omega^*(\GRS),1]$, $P^* \le \lambda^* P_1 + (1 - \lambda^*) P_2 = P_2 + \lambda^*(P_1 - P_2)$ and this right hand side is decreasing in $\lambda^*$ since $P_2 \ge P_1$.
\end{itemize}
This shows that the best strategy for the prover is to take at the first round $\lambda^* = 1 - \omega^*(\GRS)$. Notice that the above reasoning is independent on the number of remaining rounds or the number of allowed aborts. This means the same argument can be applied to each round. From there, we have that the provers optimal strategy at each round is to perform a strategy that aborts wp. $\lambda^* = 1 - \omega^*(\GRS)$. In this case, we potentially have $p^* =1$ so the provers will win all the games where they don't abort but they will most probably abort too often, for well chosen parameters. Let $P^*_{losses}(R,F)$ be the probability that the provers perform less than $F$ aborts with this strategy. We use the following Chernoff bound
\begin{proposition}[Additive Chernoff bound]
Suppose $X_1,\dots,X_n$ are independent random variables taking value in $\{0,1\}$. Let $X$ denote their sum, $p = \E[X_1]$ and $\eps > 0$. We have 
\begin{align*}\Pr[X \ge pn+ \eps n] & \le 2^{n\left((p+\eps)\log_2(\frac{p}{p+\eps}) + (1-p-\eps)\log_2(\frac{1-p}{1-p-\eps})\right)}, \\
	\Pr[X \le pn - \eps n] & \le 2^{n\left((p - \eps)\log_2(\frac{p}{p- \eps}) + (1-p+\eps)\log_2(\frac{1-p}{1-p+\eps})\right)}.
\end{align*}
\end{proposition}
In our case, the strategy of the provers outputs `Abort' wp. $\lambda^*$ and they succeed in cheating if there are at most $F = \lambda R$ aborts. We use the above Chernoff bound (second equation) with  $n = R, p = \lambda^* = 1 - \omega^*(\GRS), \eps = \lambda^* - \lambda$ which bounds the probability that there are less than $F = \lambda n$ aborts, and we write 
\begin{align}\label{Equation:D2}
P^*(R,F) & \le  2^{R\left(\lambda\log_2(\frac{\lambda^*}{\lambda}) + (1-\lambda)\log_2(\frac{1-\lambda^*}{1-\lambda})\right)} \nonumber \\
& = 2^{R\left(\lambda\log_2(\frac{1 - \omega^*(\GRS)}{\lambda}) + (1-\lambda)\log_2(\frac{\omega^*(\GRS)}{1-\lambda})\right)} 
\end{align}

\subsection{Analysis of the completeness error}
We can also use the above Chernoff bound to bound the completeness error, {\ie} the probability that the protocol aborts when all players are honest. We allow up to $F$ aborts and assume we have parameters for which the signal doesn't arrive in time with some probability $p_{\text{loss}}$. Let $CE(R,F,p_{\text{loss}})$ denote the completeness error and $\lambda = \frac{F}{R}$. Using again the Chernoff bound (first equation) with $R = n, p = p_{\text{loss}}, \eps = \lambda - p_{\text{loss}}$, we have 
\begin{align}\label{Equation:D3}CE(R,F,p_{\text{loss}}) \le 2^{R\left(\lambda\log_2(\frac{p_{\text{loss}}}{\lambda}) + (1-\lambda)\log_2(\frac{1 - p_{\text{loss}}}{1-\lambda})\right)}.
\end{align}

\vspace*{-1cm}
\section{Other parameters}\label{Appendix:OtherValues}
The plots we present in the main text correspond to $n=1704$ which allows us to have $100$ bits of security. We present here the plots for other values of $n$, for our $2$ scenarios, to present the scaling of our scheme. 

\paragraph{First scenario, different values of $n$.} $ \ $ \\ \\

	\includegraphics[width = 8cm]{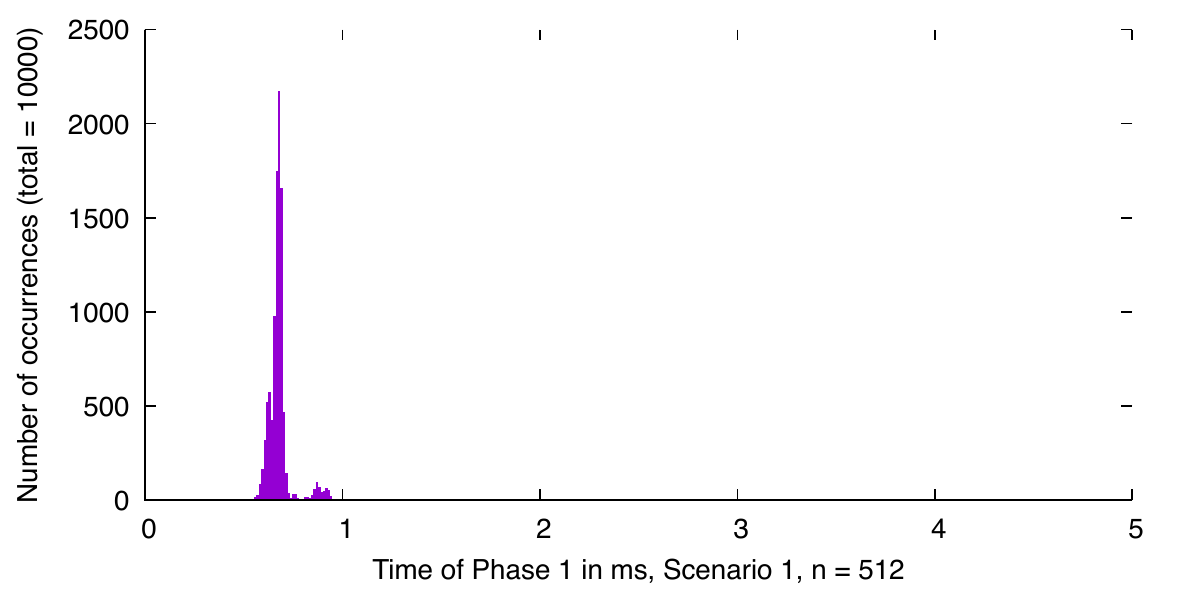}
	
	\includegraphics[width = 8cm]{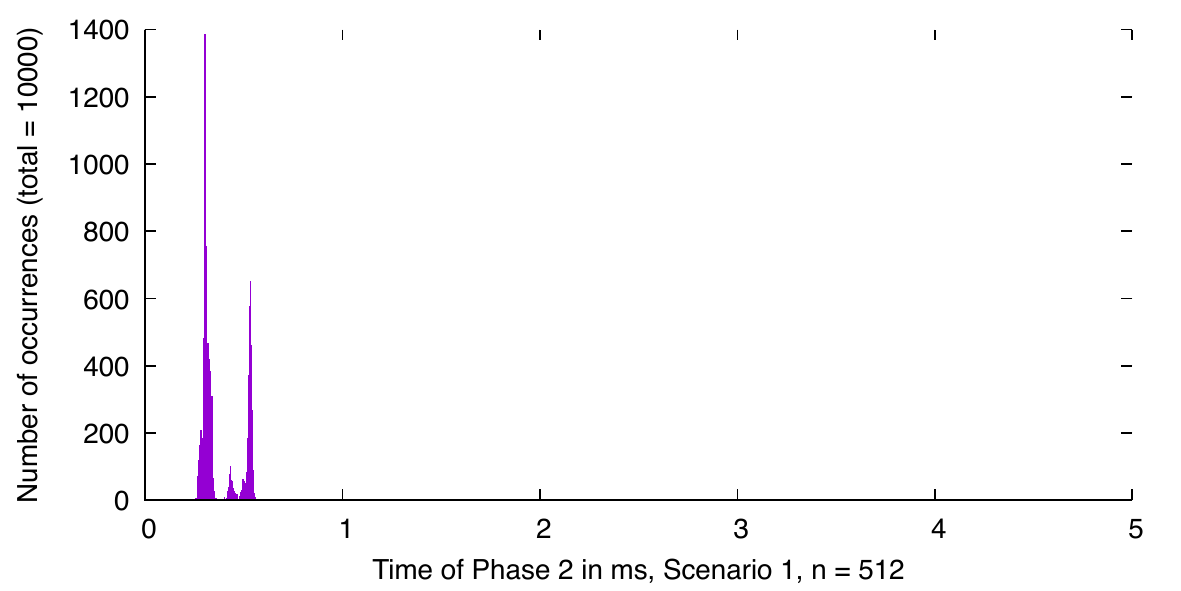} 

	\includegraphics[width = 8cm]{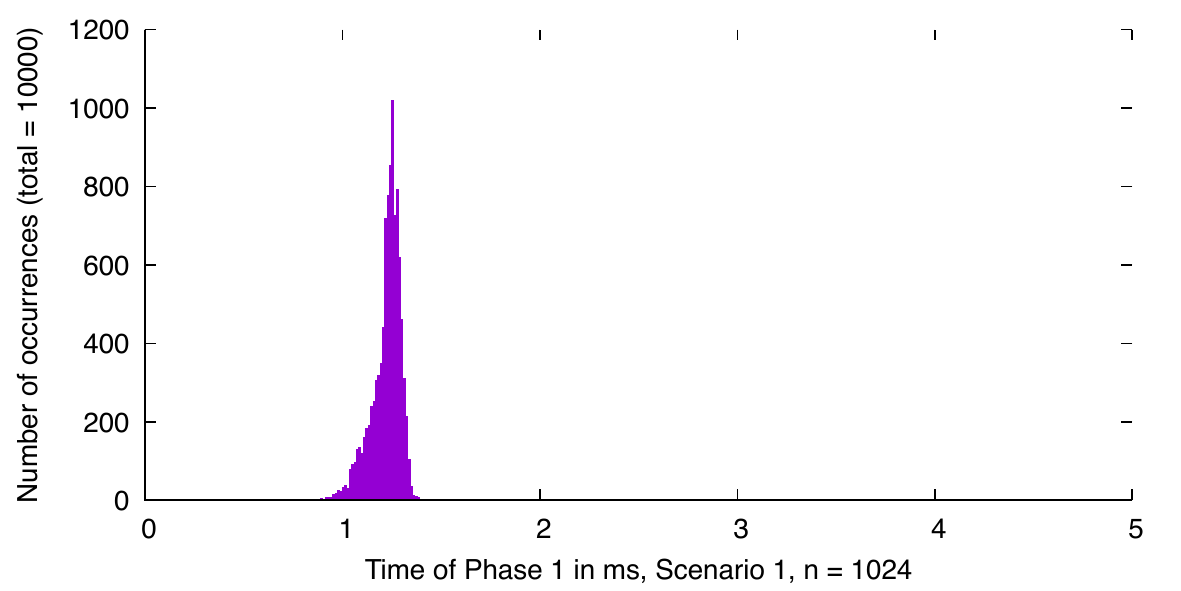}
	
   \includegraphics[width = 8cm]{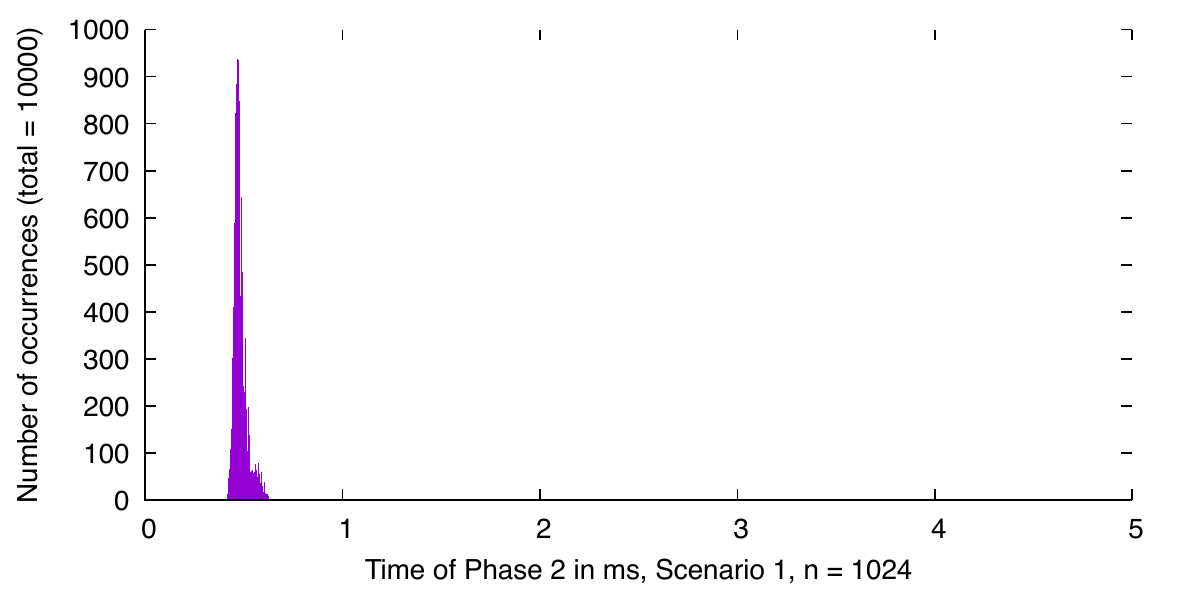} 

	\includegraphics[width = 8cm]{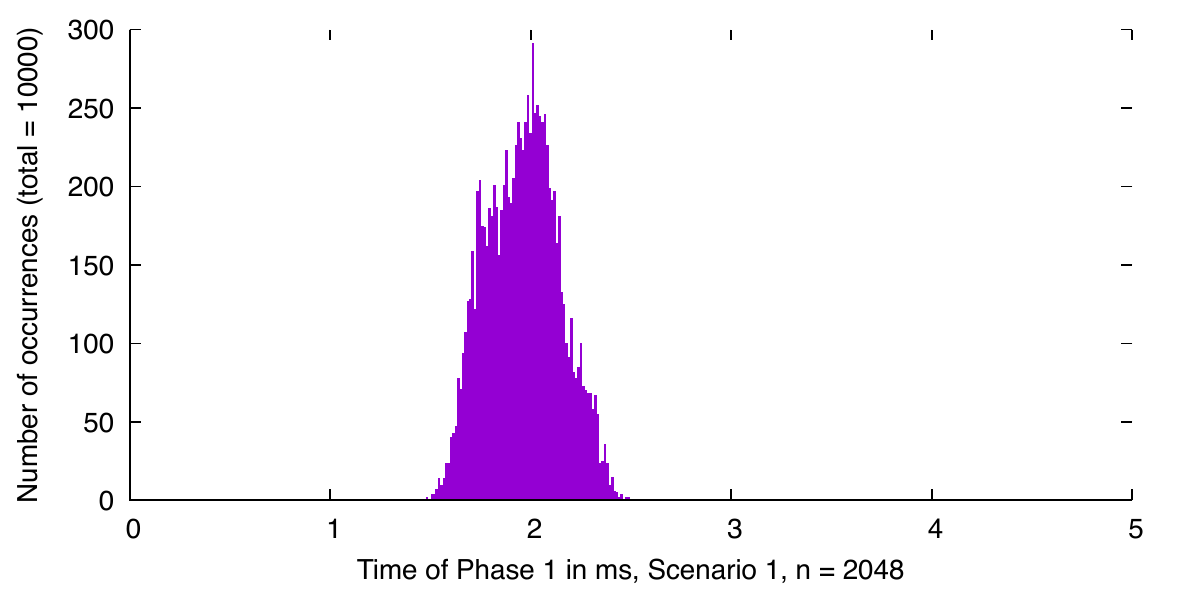}
	
	\includegraphics[width = 8cm]{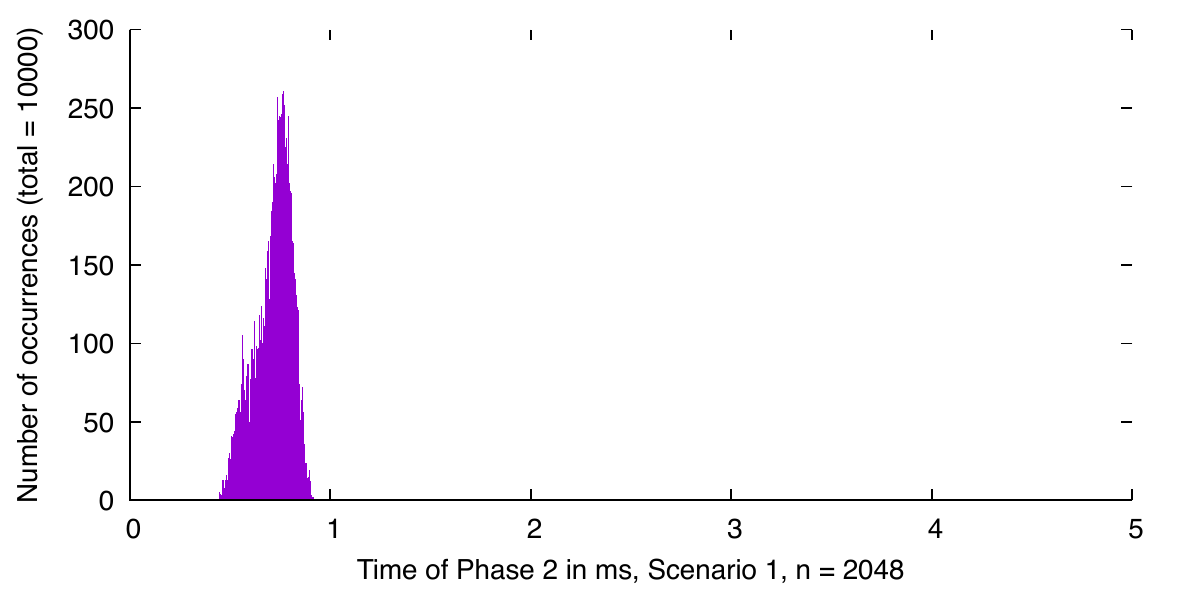}

	\includegraphics[width = 8cm]{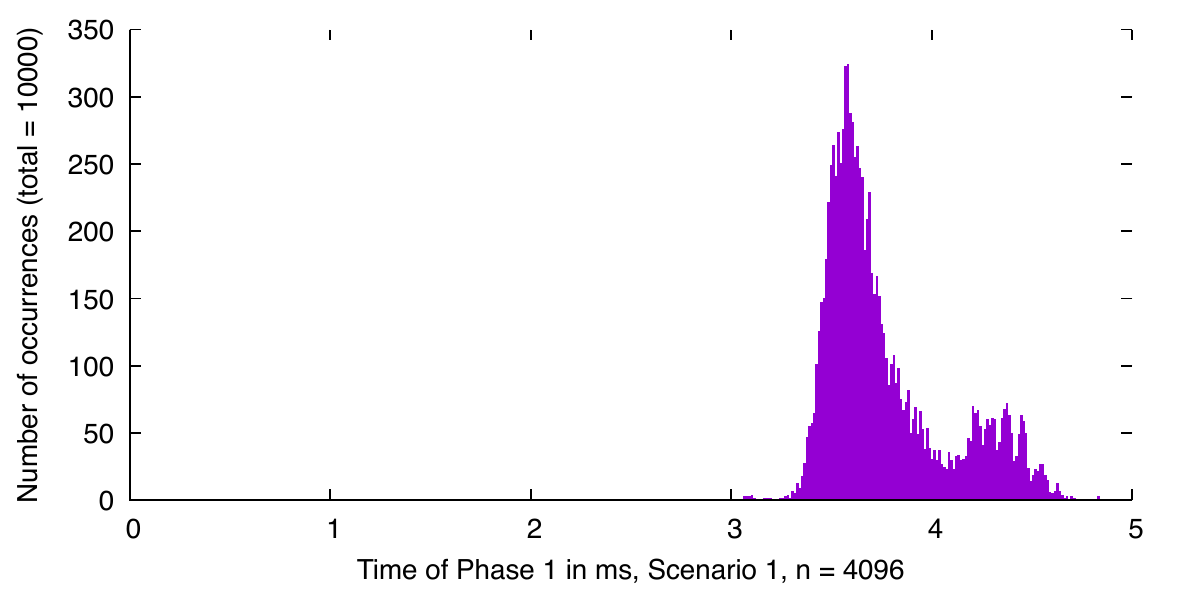}
	
	\includegraphics[width = 8cm]{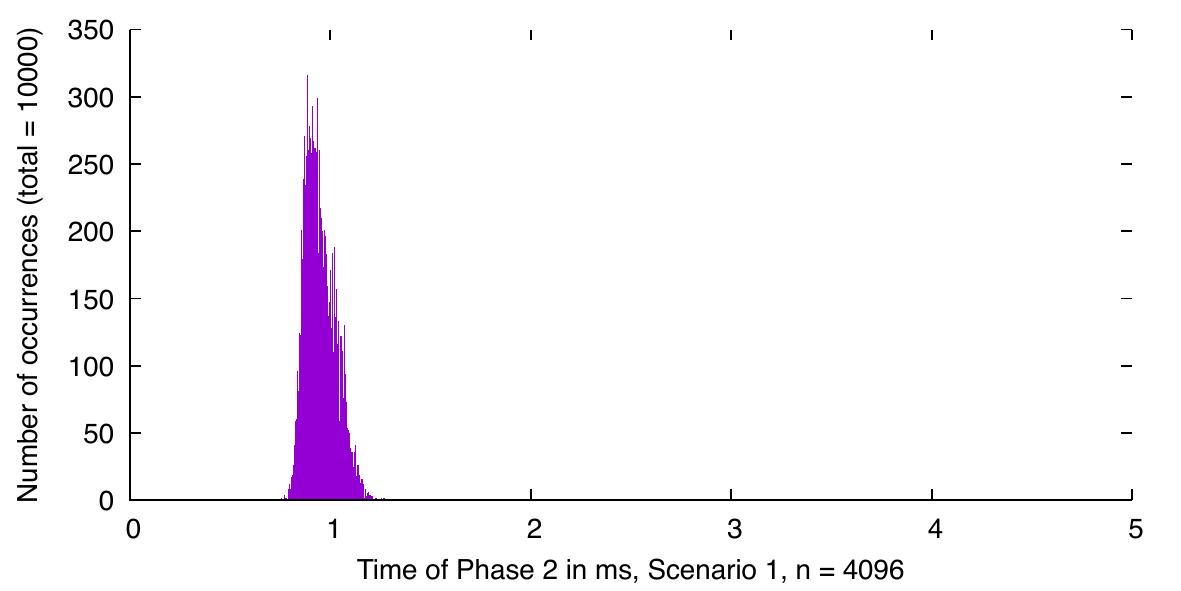}
\newpage
\paragraph{Second scenario, different values of $n$.}
$ \ $ 

\includegraphics[width = 8cm]{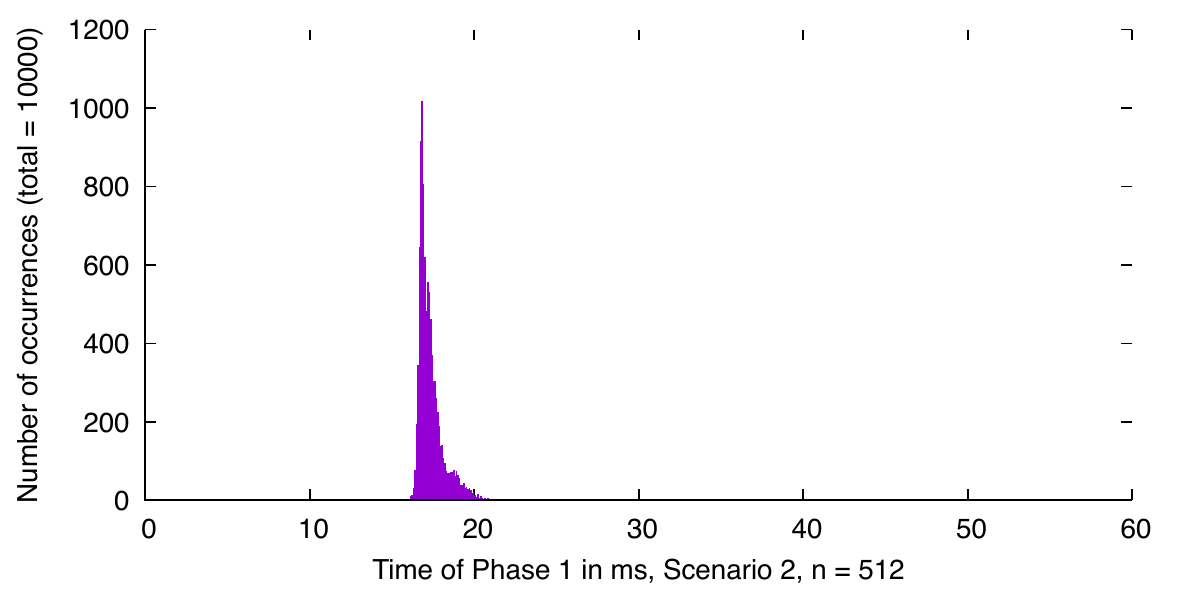}

\includegraphics[width = 8cm]{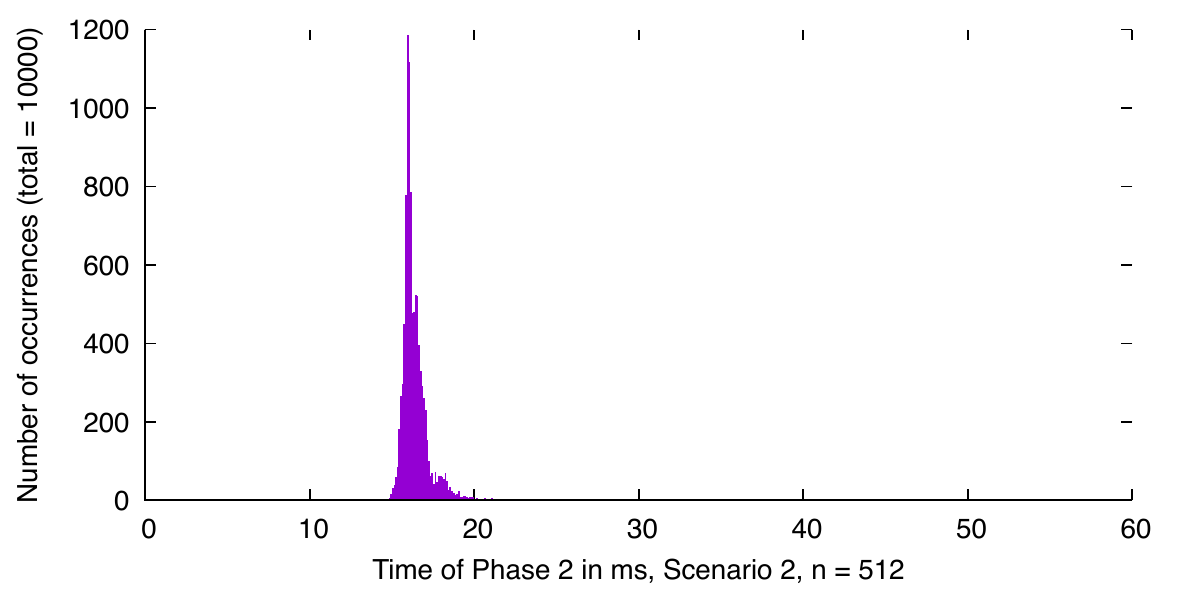} 

\includegraphics[width = 8cm]{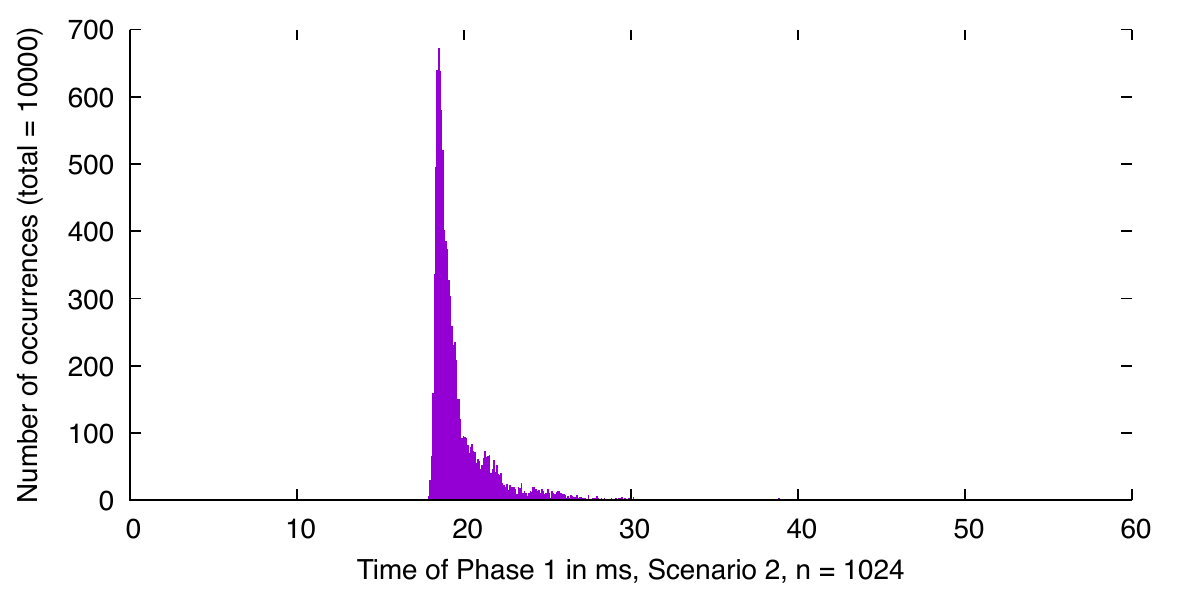}

\includegraphics[width = 8cm]{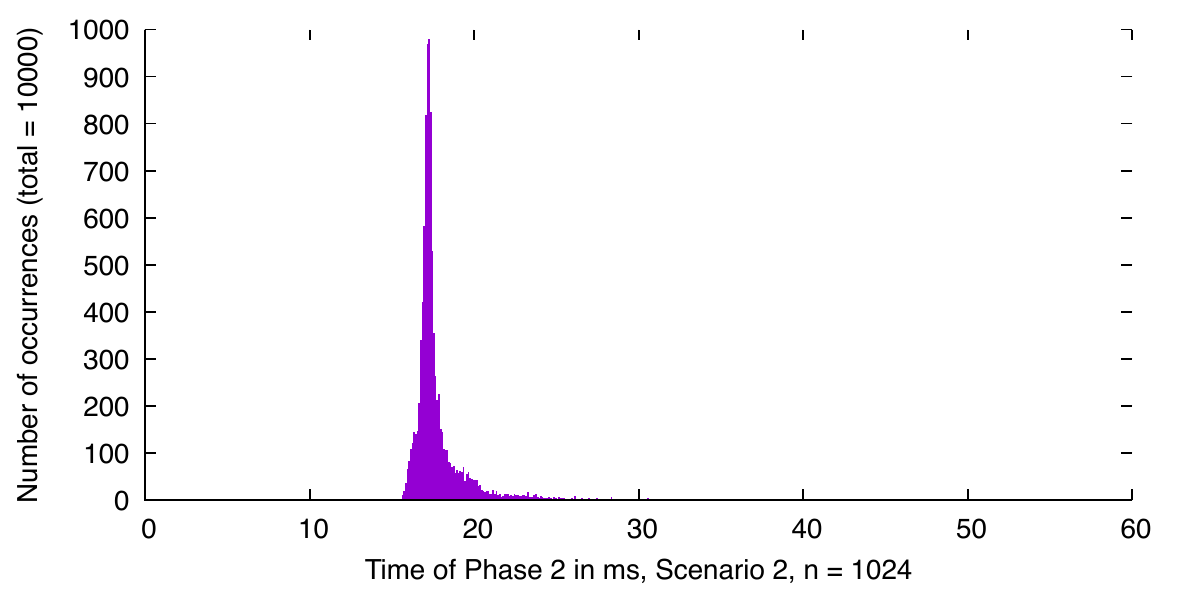} 

 \includegraphics[width = 8cm]{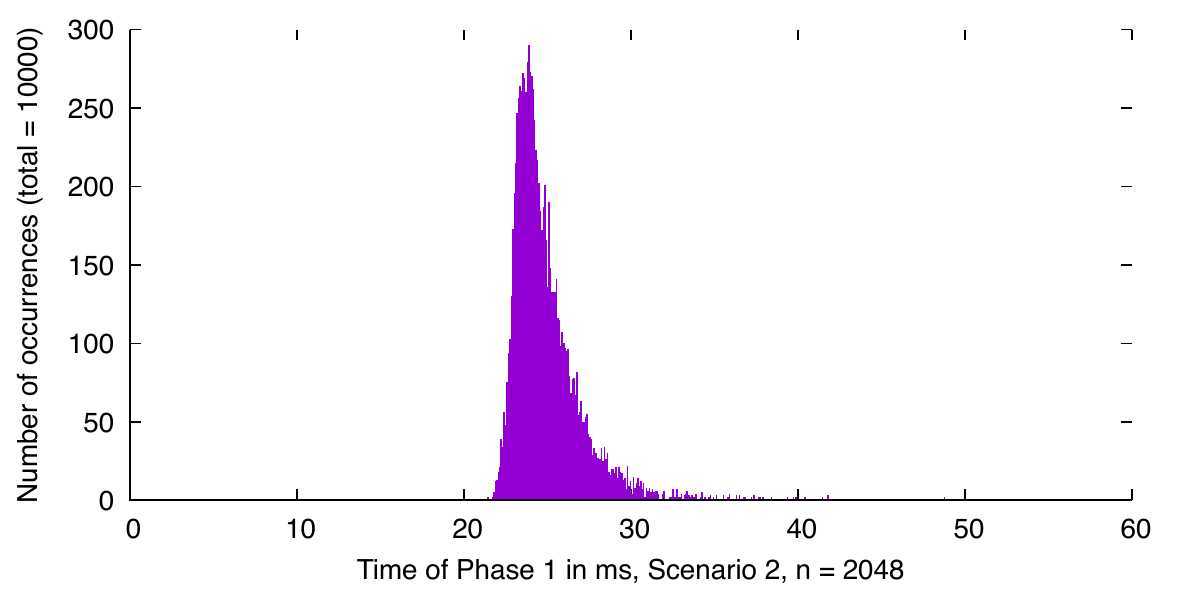} 
 
\includegraphics[width = 8cm]{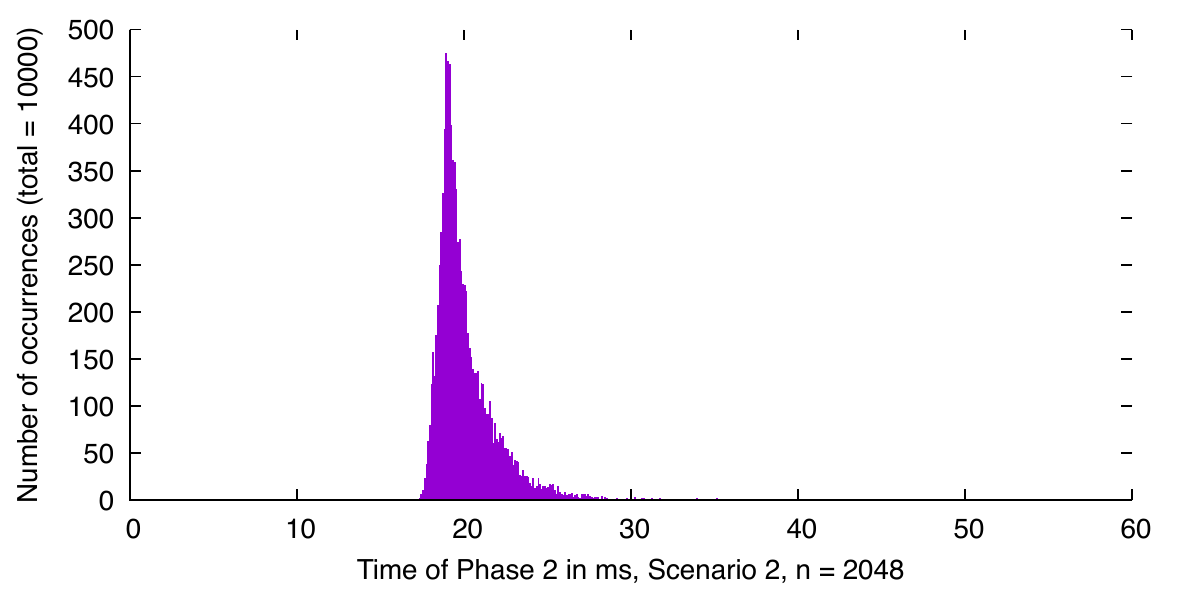} 

\includegraphics[width = 8cm]{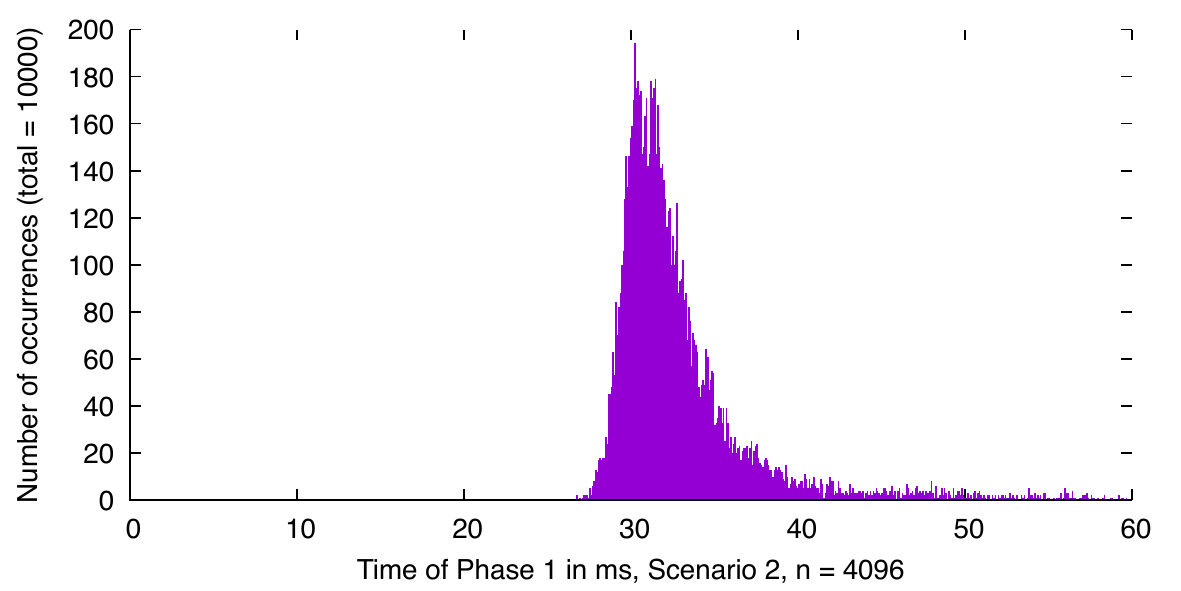}

\includegraphics[width = 8cm]{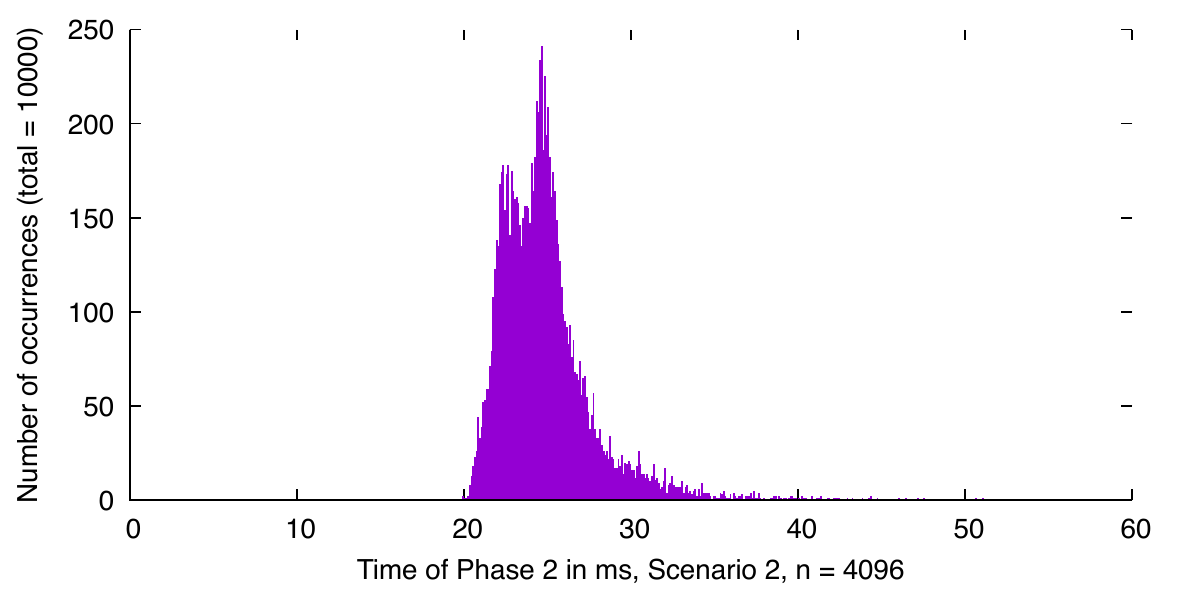}

\end{appendix}
\end{document}